\documentclass[11pt]{article}

\usepackage[final]{acl}

\usepackage{times}
\usepackage{latexsym}
\usepackage{algorithm}
\usepackage{algorithmic}
\usepackage{amsmath}
\usepackage{booktabs}
\usepackage{amssymb}
\usepackage{multirow}
\usepackage{graphicx}
\usepackage{tabularx}
\usepackage{enumitem}
\usepackage{array}
\usepackage{makecell}
\usepackage{amsthm}
\newtheorem{theorem}{Theorem}
\newtheorem{corollary}{Corollary}
\newcommand{\emap}{\kappa}

\usepackage[T1]{fontenc}

\usepackage[utf8]{inputenc}

\usepackage{microtype}

\usepackage{inconsolata}

\usepackage{graphicx}

\title{Structure Guided Retrieval-Augmented Generation for Factual Queries}

\author{
  Miao Xie\textsuperscript{1}\thanks{Corresponding authors.},
  Xiao Zhang\textsuperscript{1},
  Yi Li\textsuperscript{2},
  Chunli Lv\textsuperscript{1,3}\footnotemark[1] \\
  \textsuperscript{1}College of Information and Electrical Engineering, China Agricultural University, China \\
  \textsuperscript{2}College of Computing and Data Science, Nanyang Technological University, Singapore \\
  \textsuperscript{3}Key Laboratory of Agricultural Informatization Standardization, Ministry of Agriculture and Rural Affairs, China \\
  \texttt{0520shui@163.com, B20243080802@cau.edu.cn, liyi0067@e.ntu.edu.sg, lvcl@cau.edu.cn}
}

\begin{document}
\maketitle
\begin{abstract}
Retrieval-Augmented Generation (RAG) has been proposed to mitigate hallucinations in large language models (LLMs), where generated outputs may be factually incorrect.
However, existing RAG approaches predominantly rely on vector similarity for retrieval, which is prone to semantic noise and fails to ensure that generated responses fully satisfy the complex conditions specified by factual queries.
To address this challenge, we introduce a novel research problem, named \textsc{\textbf{E}xact \textbf{R}etrieval \textbf{P}roblem (ERP)}. To the best of our knowledge, this is the first problem formulation that explicitly incorporates structural information into RAG for factual questions to satisfy all query conditions. For this novel problem, we propose \textsc{\textbf{S}tructure \textbf{G}uided \textbf{R}etrieval-\textbf{A}ugmented \textbf{G}eneration (SG-RAG)}\footnote{Code is publicly available at: \url{https://github.com/CAU-X-AI-Lab/SG-RAG}}, 
which models the retrieval process as an embedding-based subgraph matching task, and uses the retrieved topological structures to guide the LLM to generate answers that meet all specified query conditions.
To facilitate evaluation of \textsc{ERP}, we construct and publicly release \textsc{\textbf{E}xact \textbf{R}etrieval \textbf{Q}uestion \textbf{A}nswering (ERQA)}\footnote{Dataset is publicly available at: \url{https://github.com/CAU-X-AI-Lab/ERQA}}
, a large-scale dataset comprising $120{,}000$ fact-oriented QA pairs, each involving complex conditions, spanning $20$ diverse domains.
The experimental results demonstrate that \textsc{SG-RAG} significantly outperforms strong baselines on \textsc{ERQA}, delivering absolute gains of $20.68$--$50.88$ percentage points, corresponding to $34$\%--$450$\% relative improvements across metrics, while maintaining reasonable computational overhead.

\end{abstract}

\section{Introduction}
LLMs suffer from a well-documented limitation: they often produce hallucinations that are fluent but factually incorrect~\cite{huang2025survey}. 
Such hallucinations have been observed in medical QA~\cite{pal2023med}, legal drafting~\cite{curran2023hallucination}, and scientific writing~\cite{sui2024confabulation}, where factual errors may mislead users and undermine trust~\cite{luo2024hallucination}.
RAG alleviates hallucinations by incorporating external knowledge, shows strong performance in QA and assistant tasks. Current RAG methods can be categorized into two paradigms: chunk-based RAG, represented by NaiveRAG~\cite{lewis2020retrieval}, and graph-based RAG, exemplified by GraphRAG~\cite{edge2024local}. Both paradigms have been adopted in real-world systems such as Dify~\cite{arai2024design} and LangChain~\cite{topsakal2023creating}.

\begin{figure}[!b]
\centering
\includegraphics[width=0.99\columnwidth]{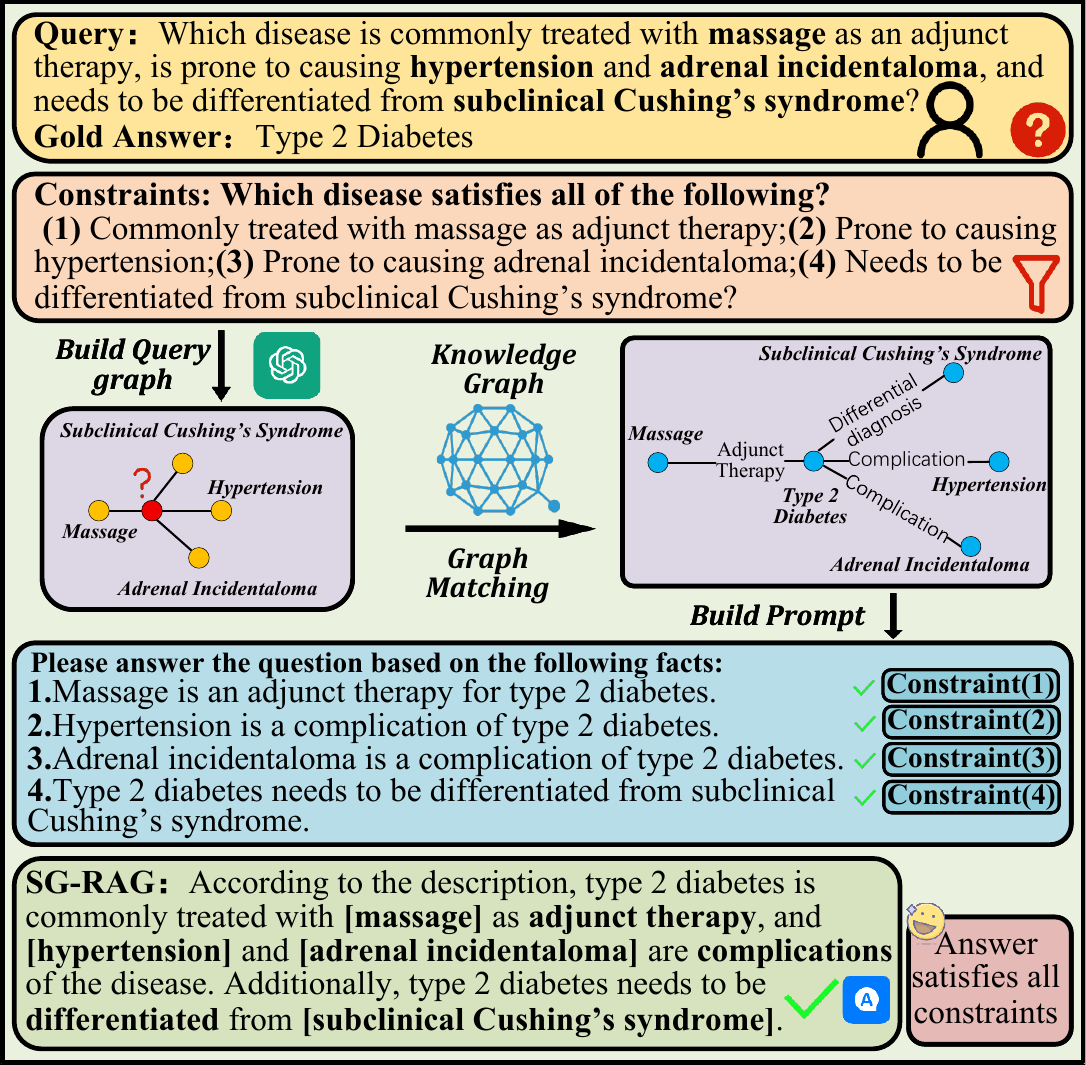} 
\caption{Multi-Condition QA Example.}
\label{fig1}
\end{figure}

However, in real-world applications, many fact-oriented queries require that answers satisfy all conditions in the query. As shown in Fig~\ref{fig1}, in a medical QA scenario, a user may ask:
\textit{``Which disease commonly uses massage as adjuvant therapy, is prone to cause hypertension and adrenal incidentaloma, and requires differential diagnosis from subclinical Cushing's syndrome?''}
This query imposes four distinct conditions on the target disease: (1) massage is used as adjuvant therapy; (2) it is prone to induce hypertension; (3) it tends to cause adrenal incidentaloma; (4) it requires differential diagnosis from subclinical Cushing's syndrome. Each condition represents a constraint, and collectively they define a complex query intent. Unlike single constraint queries, such multi-constraint queries demand that a correct answer satisfy all specified constraints.
Current RAG methods such as NaiveRAG and GraphRAG often struggle to achieve this objective. The limitation lies in their reliance on semantic similarity to rank retrieval results and generate responses based on local text chunks or subgraphs from external knowledge sources. Thus, they fail to retrieve the full set of supporting information required to satisfy all the constraints in the query, leading to incomplete answers.

To address this gap, we define a novel research problem called \textsc{\textbf{E}xact \textbf{R}etrieval \textbf{P}roblem (ERP)}. Given a user query involving multiple constraints, the goal of \textsc{ERP} is to retrieve information that precisely and comprehensively satisfies all specified conditions, from external knowledge sources, ensuring that the generated answer complies with every constraint in the query.

To solve the \textsc{Exact Retrieval Problem}, we propose \textsc{\textbf{S}tructure \textbf{G}uided \textbf{R}etrieval-\textbf{A}ugmented \textbf{G}eneration(SG-RAG)}.
Unlike existing RAG methods that rely on vector similarity, \textsc{SG-RAG} introduces a retrieval mechanism guided by structure that better respects the multi-constraint structure of the query.
The core idea is to model knowledge retrieval as a subgraph matching task based on embeddings, enabling the system to precisely and comprehensively retrieve information from external knowledge sources that satisfy the constraints specified in the query. The information of the retrieved subgraphs is then converted into prompts that guide the LLM in generating answers that satisfy all constraints. In the example of Figure~\ref{fig1}, \textsc{SG-RAG} retrieves a subgraph that satisfies all query constraints and produces the correct answer, \emph{Type 2 diabetes}.

To systematically evaluate the ability of \textsc{SG-RAG} to solve \textsc{ERP}, we construct and publicly release a large-scale benchmark dataset called \textsc{ERQA (\textbf{E}xact \textbf{R}etrieval \textbf{Q}uestion \textbf{A}nswering)}.
\textsc{ERQA} consists of three subsets built from diverse domains: (1) \textsc{FB-ERQA}, an English encyclopedic graph contributing $80{,}000$ queries; (2) \textsc{UD-ERQA}, a multi-disciplinary English dataset built from textbooks covering $18$ distinct domains, contributing $10{,}000$ queries; and (3) \textsc{CM-ERQA}, a Chinese medical knowledge graph contributing $30{,}000$ queries. Each query includes multiple constraints, with a ground-truth answer for evaluation.

Based on \textsc{ERQA}, we compare \textsc{SG-RAG} with several RAG baselines. \textsc{SG-RAG} achieves significant and consistent gains across metrics, yielding $20.68$--$50.88$ percentage-point absolute improvements, corresponding to $34$\% to over $450$\% relative gains, while maintaining reasonable computational efficiency.

Our main contributions are as follows.

\begin{itemize}[leftmargin=*, itemsep=1pt, topsep=2pt, parsep=0pt, partopsep=0pt]
    \item To address the challenges observed in real-world applications, we propose a novel RAG query paradigm, called \textsc{ERP}.
    \item To solve the \textsc{ERP}, we propose \textsc{SG-RAG}, a method that leverages an embedding-based subgraph matching mechanism.
    \item To enable a systematic evaluation of \textsc{SG-RAG}, we construct and \textbf{publicly} release \textsc{ERQA}, a benchmark containing $120,000$ factual QA pairs across $20$ domains.
    \item On \textsc{ERQA}, \textsc{SG-RAG} outperforms strong baselines by $20.68$--$50.88$ points, corresponding to $34$\% to over $450$\% relative gains across metrics.

\end{itemize}

\section{Related Work}

\textbf{Retrieval-Augmented Generation} was introduced by \cite{lewis2020retrieval} to improve QA by integrating vector similarity-based retrieval with the generation of LLMs. Subsequent research has progressed in two major directions: (i) semantic alignment controlled by the query and (ii) structure-aware retrieval with explicit knowledge modeling.
In the first line, HyDE~\cite{gao2023precise} introduced hypothetical answer generation for backward retrieval. MEMORAG~\cite{qian2024memorag} incorporated a memory module for multi-turn coherence, while RQ-RA~\cite{chan2024rq} improved multi-hop QA via structured query rewriting.
In the second line, GraphRAG~\cite{edge2024local} pioneered entity graph integration and community-based paragraph retrieval. LightRAG~\cite{guo2024lightrag} reduced the cost of graph construction through a simplified structure. HopRAG~\cite{liu2025hoprag} introduced multi-hop subgraphs for long-range access, GRAG~\cite{xu2025align} focused on dynamic graph evolution, and G-Refer~\cite{li2025g} employed dominant embeddings with contrastive learning to enforce structural consistency.
For domain-specific QA, MedRAG~\cite{zhao2025medrag} integrated medical ontologies, HippoRAG~\cite{jimenez2024hipporag} mimicked hippocampal memory encoding, and AMAR~\cite{xu2025harnessing} enabled multi-view retrieval of entities, attributes, and paths.
In summary, existing RAG methods effectively mitigate LLM hallucination, but their reliance on vector similarity limits their ability to precisely handle real-world queries with complex conditions.
\textbf{Subgraph Matching} techniques fall into join-based and backtracking-based categories.
Join-based approaches include BiGJoin and its variants~\cite{ammar2018distributed} achieving worst-case optimal joins even in dynamic graphs, fractional-cover-based joins~\cite{ngo2018worst}, and distributed matching via Timely Dataflow~\cite{lai2019distributed}. Systems such as EmptyHeaded~\cite{aberger2017emptyheaded} utilize SIMD joins with high-level queries, while cost-based optimizers refine join orders~\cite{mhedhbi2019optimizing}. Hybrid methods such as RapidMatch~\cite{sun2020rapidmatch} merge joins with exploration; SEED~\cite{lai2016scalable} applies clique/star units with bushy joins; and TwinTwigJoin~\cite{lai2015scalable} achieves optimality on MapReduce. The visual and partial matching are addressed by FERRARI~\cite{wang2020ferrari,icde-XIE} and PANDA~\cite{xie2017panda, PANDA-SYS}, respectively.
Backtracking-based methods include VF3~\cite{carletti2015vf2} and VF2 Plus~\cite{carletti2017challenging} for dense and sparse graphs, QuickSI~\cite{shang2008taming} with optimized orderings, and redundancy reduction strategies via Cartesian product postponement~\cite{bi2016efficient} or algebraic pruning~\cite{he2008graphs}. CECI~\cite{bhattarai2019ceci} leverages embedding clusters, while parallelized exploration~\cite{sun2012efficient} scales to billion-node graphs. Pruning-free methods~\cite{bonnici2013subgraph} also perform well in biological networks. Recently, GNN-PE~\cite{ye2024efficient} enables exact subgraph matching using dominant path embeddings.
Although these methods improve scalability and efficiency, they typically assume fully labeled graphs. However, in RAG, the query target is often unknown, making such methods difficult to apply.

\begin{figure*}[!h]
\centering
\includegraphics[width=0.97\textwidth]{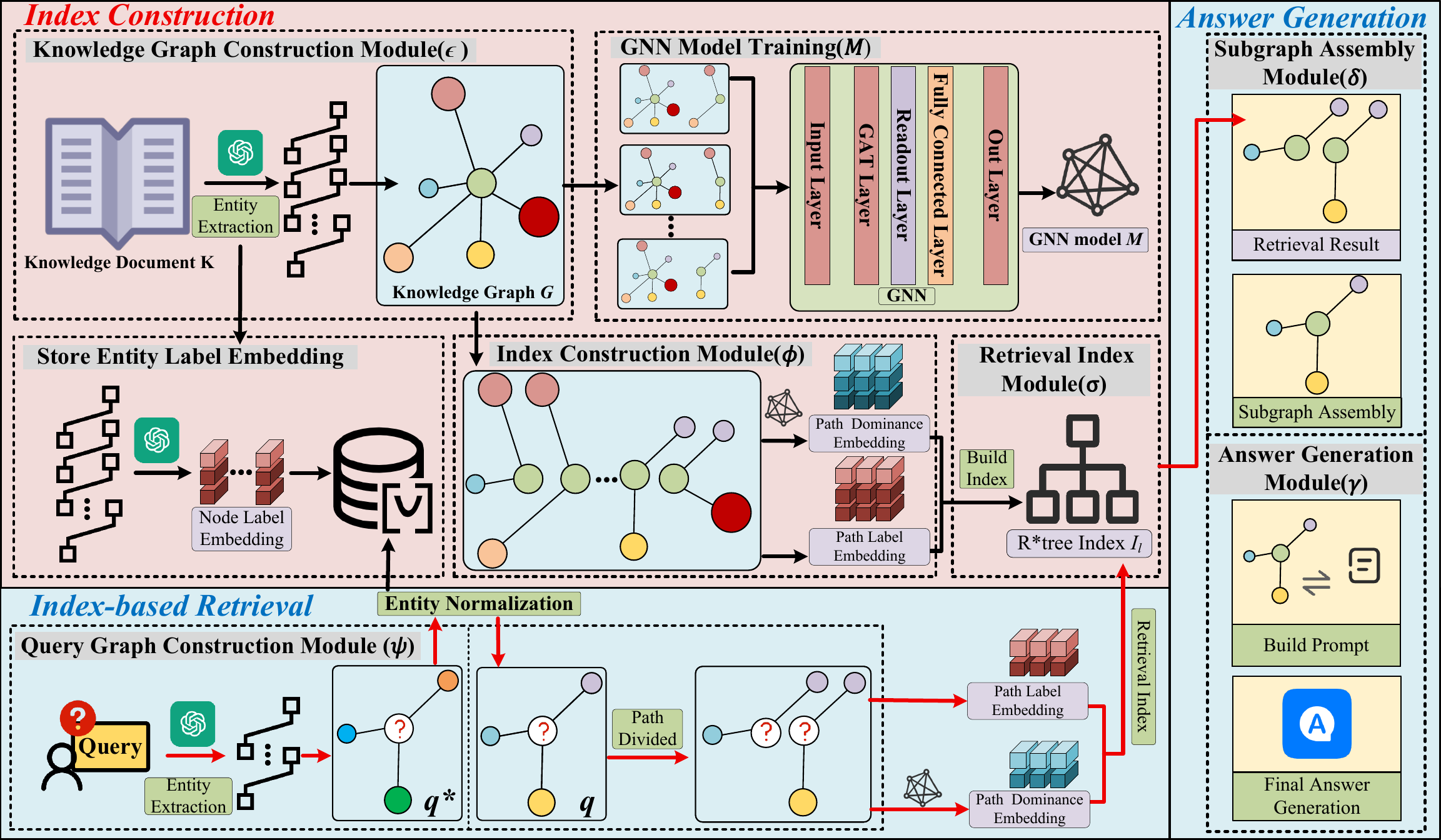} % Reduce the figure size so that it is slightly narrower than the column.
\caption{SG-RAG system architecture for structure guided retrieval and answer generation.}
\label{fig2}
\end{figure*}

\section{ Preliminaries and Problem Definition }
In this section, we present the core notations and formal definitions. Symbols are summarized in Appendix~\ref{appendix:Symbol Definitions}.

\noindent\textbf{Definition.}
A graph $G$ is a quadruple $(V, E, \emap, L)$, where: $V$ is a set of vertices, each $v_i \in V$ represented as $(\mathrm{id}(v_i), \ell(v_i), \xi(v_i))$, with unique id, semantic label, and textual description. $E$ is a set of directed edges $e_{ij} = (v_i, v_j)$, each encoded as $(\mathrm{id}(e_{ij}), v_i, v_j, \xi(e_{ij}))$. $\emap: V \times V \rightarrow E$ maps a pair of vertices to their edge. $L$ is a labeling function assigning each vertex $v_i \in V$ a label $\ell(v_i)$. Given two graphs $G_A$ and $G_B$, if they are isomorphic, we write $G_A \equiv G_B$~\cite{10.1145/2463676.2465300}.

\noindent\textbf{Problem Definition.}
\noindent\textsc{Exact Retrieval Problem (ERP):}
Given: a natural language query $q^{o}$ with a set of constraints $C = \{c_1, ..., c_k\}$ ($k \geq 2$); a knowledge corpus ${K} = \{k_1, ..., k_n\}$ and a large language model $\Lambda$, \textsc{ERP} aims to retrieve from ${K}$ the most accurate knowledge that satisfies all the constraints in $C$, and use it to prompt the $\Lambda$ to generate a factual answer. Each constraint $c_i$ typically corresponds to an entity related to the target answer. When such precise information cannot be retrieved, the goal is to identify relevant evidence to guide the $\Lambda$ in generating a reliable answer.

\section{Method}
\subsection{System Overview}
To address \textsc{ERP}, $q^{o}$ can be structured into a query graph $q$, and $K$ can be represented as a knowledge graph $G$. This enables \textsc{ERP} to be transformed into a subgraph matching task: Find a subgraph $g \subseteq G$ such that $g \equiv q$, ensuring that all query constraints are structurally satisfied.
However, since the subgraph isomorphism is \textbf{NP-complete}, applying it effectively in the RAG setting poses a significant challenge. 
To overcome this challenge, we propose \textsc{SG-RAG}. As shown in Figure~\ref{fig2}, we formalize \textsc{SG-RAG} as a tuple:
\begin{equation}
\small
    \text{SG-RAG} = (\epsilon, \mathcal{M}, \phi, \psi, \sigma, \delta, \gamma)
\end{equation}
where $\epsilon$: document structuring module that converts knowledge corpus $K$ into a knowledge graph $G$; $\mathcal{M}$: GNN model for generating path dominance embeddings; $\phi$: index construction module that builds R*-Tree $I_l$ over path embeddings; $\psi$: query graph construction module that extracts, normalizes, completes and decomposes $q^{o}$; $\sigma$: path-level retrieval module using $I_l$; $\delta$: subgraph assembly module that forms exact subgraphs; $\gamma$: answer generation module that prompts LLM with retrieved subgraphs.
The overall system execution is defined:

\begin{equation}
\small
    \text{SG-RAG}(q^{o}; K) = \gamma(q^{o}, \delta(\psi(q^{o})), \sigma(\phi(\mathcal{M}, \epsilon(K))))
\end{equation}
\textsc{SG-RAG} operates in three stages, as illustrated in Algorithm~\ref{alg:sg-rag-framework}:
\textbf{Index Construction} (Lines 1–6): the document corpus $K$ is converted into a knowledge graph $G$ via $\epsilon$. Label embeddings of entity and path $o_0(v_i)$ and $o_0(p_z)$ are calculated, and the GNN model $\mathcal{M}$ is trained. The path dominance embeddings $o(p_z)$ are computed and indexed by $\phi$.
\textbf{Index-based Retrieval} (Lines 7–16): the user query $q^{o}$ is parsed into a query graph $q^*$ with label embeddings $o_0(q_i)$. Entity normalization, path decomposition, and label completion are performed by $\psi$. Each completed query path set $Q' \in P$ is embedded in $\mathcal{M}$, and exact path candidates are retrieved using $\sigma$.
\textbf{Answer Generation} (Lines 17–27):  subgraphs are assembled by $\delta$, and used by $\gamma$ to generate the answer $a$ based on the best available subgraph.

\begin{algorithm}[!h]
\small
\caption{SG-RAG Framework}
\label{alg:sg-rag-framework}
\begin{algorithmic}[1]
\REQUIRE External documents $K$, User query $q^{o}$
\ENSURE Answer $a$
\STATE Extract knowledge graph $G \leftarrow \epsilon(K)$
\STATE Generate label embeddings $o_0(v_i)$ for all $v_i \in V(G)$ via LLM
\STATE Generate path label embeddings $o_0(p_z)$ for all paths $p_z$ of length $l$
\STATE Train GNN model $\mathcal{M}$ 
\STATE Compute path dominance embeddings $o(p_z)$ using $\mathcal{M}$
\STATE Build R*-Tree index $I_l$ over $o_0(p_z)$ and $o(p_z)$
\STATE Extract query graph $q^*$ from $q^{o}$ and generate $o_0(q_i)$ for all $q_i \in V(q^*)$ via LLM
\STATE Normalize entities in $q^*$ to obtain $q$
\STATE Decompose $q$ into path set $Q$ of length $l$
\STATE Compute $o_0(p_q)$ for each $p_q \in Q$ 
\STATE Predict possible labels using $I_l$
\STATE Complete $Q$ into fully labeled path sets $P$ 
\STATE $S \leftarrow \emptyset$
\FOR{each query path set $Q' \in P$}
    \STATE Compute $o(p_q)$ for each $p_q \in Q'$ using $\mathcal{M}$
    \STATE Retrieve $\text{cand\_list}$
    \STATE Assemble $g$ from $\text{cand\_list}$, add to $S$
\ENDFOR

\IF{$S \ne \emptyset$}
    \STATE Generate $a$ using information of $g \in S$ as prompt
\ELSE
    \STATE Construct fallback subgraph $g''$ by retrieving 1-hop neighbors of known entities in $q$
    \STATE Generate $a$ using information of $g''$ as prompt
\ENDIF
\RETURN $a$
\end{algorithmic}
\end{algorithm}

\subsection{Index Construction}

This stage corresponds to the components $\epsilon$, $\mathcal{M}$, and $\phi$ in the formal definition of the \textsc{SG-RAG}. 

\paragraph{Generation of the Node and Path Label Embedding.}
\textsc{SG-RAG} first converts the raw corpus $K$ into a structured knowledge graph $G$, where semantically meaningful entities are represented as nodes and explicit relations as edges (see Appendix~\ref{appendix:Graph Construction} for details). The node label embeddings are then computed using a pre-trained LLM. For each node $v_i \in V(G)$ with the label $\ell(v_i)$, the embedding is defined as $o_0(v_i) = \text{LLM}(\ell(v_i)) \in \mathbb{R}^F$. 
To represent the semantics of the path, we define the label embedding of a path $p_z = [v_1, v_2, \dots, v_k]$ as the concatenation of its embeddings of the constituent node: $o_0(p_z) = [o_0(v_1), o_0(v_2), \dots, o_0(v_k)] \in \mathbb{R}^{k \cdot F}$. It preserves both semantic content and node order for alignment during path-level retrieval.

\paragraph{Training of the Dominant Embedding Model.}
To enable exact subgraph matching, \textsc{SG-RAG} introduces a dominant embedding mechanism based on GNN. This mechanism learns structural representations of each node and its surrounding subgraph. For each node $v_i$, a 1-hop star subgraph $g_{v_i}$ is constructed and encoded using a Graph Attention Network (GAT)-based architecture. The final dominant embedding $o(v_i)$ captures the topological context around $v_i$.
To ensure structural containment, we enforce a dominance constraint: for any suitable substructure $s_{v_i} \subset g_{v_i}$, we require $o(s_{v_i}) \preceq o(g_{v_i})$. To implement this constraint during training, we introduce the following loss function:
\begin{equation}
\small
    \mathcal{L} = \sum \left\| \max(0, o(g_{v_i}) - o(s_{v_i})) \right\|_2^2
\end{equation}
This loss penalizes any violation of the dominance condition, encouraging the GNN to learn embeddings where substructures are embedded in semantically smaller regions than their supersets.
For a path $p_z = [v_1, \dots, v_k]$, the dominant embedding at the path level is computed as the concatenation $o(p_z) = [o(v_1), \dots, o(v_k)]$. These embeddings are used for pruning and matching: a query path $p_q$ is considered a match if $o(p_q) \preceq o(p_z)$ in all dimensions, indicating that the candidate structurally contains the query.
This embedding formulation enables efficient structure-aligned filtering in the retrieval phase. See Appendix~\ref{appendix:gnn-dominant} for full training details and architectural illustration.

\paragraph{Construction of the Path Index.}
To enable path-based subgraph retrieval, \textsc{SG-RAG} constructs an R*-Tree index over path embeddings.
\begin{equation}
\small
    I_l = \text{R*-Tree}\left(\left\{ o_0(p_z),\; o(p_z)\;|\; |p_z| = l \right\}\right),
\end{equation}
where each path $p_z$ of length $l$ is encoded by both its semantic embedding $o_0(p_z)$ and its structural embedding $o(p_z)$.
These embeddings are derived as follows: $o_0(p_z)$ is the concatenation of the node label embeddings to capture the semantics of the path, while $o(p_z)$ is formed by concatenating the dominant embeddings of all nodes along the path $p_z$, where the dominant embedding of each node is calculated using the pre-trained GNN model $\mathcal{M}$. The resulting embedding pairs $(o_0(p_z), o(p_z))$ are indexed for efficient retrieval.
Each index node stores different content based on type: \textbf{Leaf node:} stores $o_0(p_z)$ and $o(p_z)$;
\textbf{Non-leaf node:} stores minimum bounding rectangles ($MBRs$) over $o_0(p_z)$ and $o(p_z)$ for semantic and structural filtering, respectively.
At query time, the index is traversed in a heap-based best-first manner. Nodes whose $MBRs$ do not overlap with the query embedding region are pruned early, accelerating candidate filtering.

\subsection{Index-Based Retrieval}

This stage corresponds to the components $\psi$ and $\sigma$ in the formal definition of the \textsc{SG-RAG}. 
\paragraph{Query Graph Extraction and Entity Normalization.}
As part of the $\psi$ module, \textsc{SG-RAG} first transforms a query in natural language $q^{o}$ into a structured query graph $q^*$ using an LLM to extract entities and relations. It also computes the label embeddings $o_0(q_i^*) = \text{LLM}(\ell(q_i^*)) \in \mathbb{R}^F$ for each query node $q_i^*$.
To ensure consistency with the knowledge graph $G$, \textsc{SG-RAG} performs entity normalization using FAISS~\cite{douze2024faiss}: for each $q_i^*$, the most semantically similar node $q_i \in V(G)$ is retrieved and used to replace $q_i^*$, producing a normalized query graph $q$. Prompt and normalization algorithm are in Appendix~\ref{appendix:entity-normalization}.

\paragraph{Path Decomposition and Label Completion.}
To enable path-level retrieval, \textsc{SG-RAG} decomposes the query graph $q$ into a set of linear query paths and completes unknown node labels to produce fully specified path sets $P$.
We employ a cost-aware decomposition algorithm to iteratively select query paths of length $l$ that minimize edge overlap and retrieval cost. Each path is scored using a degree-based path weight and the optimal path set $Q$ is selected to fully cover $E(q)$.
For paths containing unknown nodes, \textsc{SG-RAG} performs wildcard-based label completion by traversing the R*-Tree index $I_l$ to identify candidate labels, generating a mapping $U$ from unknown nodes to candidate label sets.
All possible label combinations from $U$ are instantiated in $Q$ to form completed query path sets $P = \{ Q_1', Q_2', \ldots, Q_n' \}$ for downstream matching.
The complete procedures are in Appendix~\ref{appendix:path-completion}.

\paragraph{Path-level Retrieval.}
The module $\sigma$ performs path-level retrieval on the $I_l$ for each fully labeled query path $p_q \in P$.
Each query path is encoded with: a semantic embedding $o_0(p_q)$ through the LLM; a structural embedding $o(p_q)$ via the GNN model $\mathcal{M}$.
\textsc{SG-RAG} uses a heap-based best-first traversal strategy over $I_l$. During traversal:
\textbf{At Non-leaf nodes}, for each query path $p_q$ and each child node $N_i$, the system checks two constraints: \textbf{Semantic constraint}: whether $o_0(p_q)$ intersects with the label MBR $\text{MBR}_0(N_i)$; \textbf{Structural constraint}: whether the dominance region
    \begin{equation}
    \small
            \mathrm{DR}(o(p_q)) = \{ z \in \mathbb{R}^d \mid o(p_q)[i] \leq z[i], \forall i \}
    \end{equation}
    overlaps with $\text{MBR}(N_i)$.
Only when both constraints are satisfied is $p_q$ forwarded to $N_i$.
\textbf{At leaf nodes}, for each stored path $p_z$, it is added to the result based on \textbf{Exact match}: if $o_0(p_q) = o_0(p_z)$ and $o(p_q) \preceq o(p_z)$.
This dual filtering strategy ensures both semantic alignment and structural inclusion. All valid candidates are accumulated in $p_q.\text{cand\_list}$. The algorithm is in Appendix~\ref{appendix:exact subgraph matching algorithm}.

\subsection{Answer Generation}
\paragraph{Subgraph Assembly.}
The module $\delta$ composes candidates for full subgraphs from path-level matches to produce a set: exact matches $S$ that are isomorphic to the query graph.
For exact matches, \textsc{SG-RAG} enumerates all combinations of candidate paths from $p_q.\text{cand\_list}$ and verifies whether they can be merged into a conflict-free subgraph. 
This process enables \textsc{SG-RAG} to reconstruct structurally consistent subgraphs that satisfy all query constraints.
The assembly logic is in Appendix~\ref{appendix:assembly algorithm}.

\paragraph{Generate Answer.}
The module $\gamma$ generates answers based on matched subgraphs and supports two strategies.
If the exact match set $S$ is nonempty, \textsc{SG-RAG} extracts semantic labels and descriptions from all subgraphs $g \in S$ and combines them with the original query $q^{o}$ to construct a structured prompt for the LLM, ensuring all constraints are satisfied.
When exact subgraphs are not found, \textsc{SG-RAG} falls back to building a 1-hop neighborhood subgraph $g''$ around entities mentioned in $q^{o}$ as the context of last resort.
This ensures \textsc{SG-RAG} can still produce informative responses.
The prompt construction process is in Appendix~\ref{appendix:subgraph-inference-prompt}.

\subsection{Time Complexity Analysis}
We analyze the time complexity of the \textsc{SG-RAG} retrieval phase, which consists of three components: entity normalization, unknown label completion, and exact subgraph matching.
The overall retrieval complexity is the following.

\begin{equation}
\small
\begin{aligned}
O\Bigg(
& |V(q)| \cdot \log N + \sum_{i=0}^{h} |Q_u| \cdot f^{\,h-i+1} \cdot (1 - PP_i) \\
& + \Bigl( \prod_{j=1}^{u} t_j \Bigr) \cdot
    \sum_{i=0}^{h} |Q| \cdot f^{\,h-i+1} \cdot (1 - PP_i)
\Bigg).
\end{aligned}
\end{equation}

where $|V(q)|$ is the number of nodes in the query graph, $N$ is the number of entities in the knowledge graph, $u$ is the number of unknown nodes, $t_j$ is the candidate label count for node $j$, and $h, f, PP_i$ denote the height, fan-out, and pruning ratio of the R*-Tree.
In practice, most queries involve only one unknown node with a small number of candidates ($u=1$, $t_1 \leq 10$), simplifying the complexity to:
\begin{equation}
\small
    O\bigg( |V(q)| \cdot \log N \\ + (t_1 + 1) \cdot \sum_{i=0}^{h} |Q| \cdot f^{h - i + 1} \cdot (1 - PP_i) \bigg)
\end{equation}

This enables near-linear scalability, with modular execution and efficient index pruning ensuring low overhead. See Appendix~\ref{appendix:CA} for proof details.
\begin{table}[!b]
\centering
\small
\caption{Human validation results on sampled ERQA}
\setlength{\tabcolsep}{4pt}
\begin{tabular}{lccc}
\toprule
\textbf{Split} & \makecell[c]{\textbf{Fluency}\\\textbf{(1--5)}} & \makecell[c]{\textbf{Answerable}\\\textbf{(\%)}} & \makecell[c]{\textbf{Ambiguity}\\\textbf{(\%) $\downarrow$}} \\
\midrule
Chinese & 4.7 & 98.0 & 2.0 \\
English & 4.6 & 98.4 & 1.6 \\
Overall & 4.65 & 98.2 & 1.8 \\
\bottomrule
\end{tabular}
\label{tab:human_validation_erqa}
\end{table}

\begin{table*}[!h]
\centering
\small
\caption{Performance Comparison of SG-RAG and Baseline Methods on ERQA Subsets}
\setlength{\tabcolsep}{1pt}
\resizebox{\textwidth}{!}{%
\begin{tabular}{lccccccccccccccc}
\toprule
\multirow{2}{*}{\textbf{}} & \multicolumn{5}{c}{\textbf{FB-ERQA}} & \multicolumn{5}{c}{\textbf{CM-ERQA}} & \multicolumn{5}{c}{\textbf{UD-ERQA}} \\
\cmidrule(lr){2-6} \cmidrule(lr){7-11} \cmidrule(lr){12-16}
 & Hit@1 & Precision & Recall & F1 & Emp.S & Hit@1 & Precision & Recall & F1 & Emp.S & Hit@1 & Precision & Recall & F1 & Emp.S\\
\midrule
GPT-5.1       & 19.1\%  & 17.8\% & 19.2\% & 18.5\% & 0.412& 6.5\%   & 15.3\%  & 6.5\%   & 9.1\% & 0.192& 10.0\% & 8.4\%  & 11.1\% & 9.6\% & 0.287 \\
NaiveRAG          & 14.8\%  & 14.5\% & 15.1\% & 14.8\% &0.781& 14.2\%  & 8.2\%  & 17.2\% & 11.1\% &0.564& 14.4\% & 12.2\% & 15.7\% & 13.7\% & 0.681\\
RAPTOR            & 61.1\% & 60.5\% & 61.1\% & 60.8\% &0.128& 9.5\%   & 8.7\%  & 9.5\%   & 9.1\% &0.095& 10.7\% & 9.7\%  & 11.4\% & 10.5\% & 0.115\\
GraphRAG          & 61.8\%  & 61.7\% & 61.8\% & 61.8\% &1.432& 20.2\%  & 12.7\% & 22.3\% & 16.2\% & 1.007 & 18.8\% & 14.1\% & 20.0\% & 16.5\% &1.204\\
LightRAG          & 20.3\%  & 20.0\% & 20.4\% & 20.2\% & 0.934& 14.1\%  & 8.3\%  & 17.0\%  & 11.2\% &0.734& 17.6\% & 14.0\% & 18.8\% & 16.1\% &0.795\\
SubgraphRAG          & 33.4\%  & 33.6\% & 33.8\% & 33.7\% & 1.323& 30.1\%  & 29.8\%  & 32.9\%  & 31.3\% &1.210& 35.7\%  & 35.9\% & 36.1\% & 36.0\% & 1.273\\
HyperGraphRAG          & 30.6\%  & 31.4\% & 31.5\% & 31.4\% & 1.005& 21.4\%  & 17.3\%  & 23.9\%  & 20.1\% &0.986& 28.3\%  & 28.8\% & 29.1\% & 28.9\% & 0.996\\
LinearRAG          & 25.2\%  & 25.7\% & 26.1\% & 25.9\% & 0.867& 15.9\%  & 16.2\%  & 18.3\%  & 17.2\% &0.831& 23.6\%  & 23.5\% & 23.8\% & 23.6\% & 0.912\\
DyPRAG          & 28.5\%  & 29.0\% & 29.2\% & 29.1\% & 0.949& 19.1\%  & 20.5\%  & 21.9\%  & 21.2\% &0.908&  21.5\%  & 21.1\% & 22.9\% & 22.0\%& 0.909\\
KAG          & 40.7\%  & 40.7\% & 40.8\% & 40.8\% & 1.396&30.2\%  & 30.7\%  & 30.2\%  & 30.4\% &1.407& 24.8\%  & 25.1\% & 25.0\% & 25.0\% & 1.278\\
\textbf{SG-RAG} & \textbf{82.5\%} & \textbf{82.4\%} & \textbf{82.6\%} & \textbf{82.5\%} &\textbf{1.855}& \textbf{61.1\%} & \textbf{60.5\%} & \textbf{61.2\%} & \textbf{60.8\%} & \textbf{1.641}&\textbf{61.8\%} & \textbf{65.0\%} & \textbf{67.5\%} & \textbf{66.2\%} &\textbf{1.611} \\
\bottomrule
\end{tabular}%
}
\label{tab:erqa-comparison}
\end{table*}

\section{Experimental Study}
\subsection{Dataset}

Since most existing open RAG and QA benchmarks target single constraint queries (e.g., Natural Questions (NQ)~\cite{kwiatkowski2019natural}) and rarely verify constraint satisfaction with verifiable evidence, we construct and publicly release ERQA to systematically evaluate SG-RAG on ERP.
\textsc{ERQA} comprises three subsets: \textsc{FB-ERQA} contains $80{,}000$ English QA pairs derived from FB15K-237 with $6.1$ constraints per query on average; \textsc{UD-ERQA} contains $10{,}000$ QA pairs from a cross-domain academic corpus spanning $18$ disciplines with $4.7$ constraints per query on average; \textsc{CM-ERQA} contains $30{,}000$ Chinese QA pairs from the CPubMed-KG~\cite{zhang2025much} with $5.4$ constraints per query on average. 
To further assess the realism of ERQA queries, we conduct a human validation study on $2,000$ randomly sampled queries. Three expert annotators, including one PhD in Nutrition and two PhDs in Computer Science, independently evaluate each query in terms of fluency, answerability, and ambiguity. For cases that are difficult to judge, the annotators verify them through literature search. The results (Table~\ref{tab:human_validation_erqa}) show high fluency , high answerability , and low ambiguity, suggesting that \textsc{ERQA} queries are generally natural, understandable, and answerable. Detailed statistics, example queries, and the question-generation prompt are provided in Appendix~\ref{appendix:erqa} and Appendix~\ref{appendix:question-generation-prompt}.
We evaluate \textsc{SG-RAG} on \textsc{ERQA} for \textsc{ERP} and additionally on NQ dataset to verify its generalization to single-constraint queries.

\subsection{Experimental Setting}

\textbf{Environment.} All experiments were conducted on a local workstation running Ubuntu $22.04$ LTS. The hardware configuration includes a $16$ core, $32$ thread Intel Core\texttrademark~CPU and an NVIDIA GeForce RTX $4060$ GPU (driver version $560.94$, CUDA $12.6$).
\textbf{Model Configuration.} During dataset construction, we employed the \texttt{GLM-4-Flash} to generate structured queries.
% For all vector embedding tasks, we used  \texttt{text-embedding-3-small}  to maintain consistency between methods.
For all vector embedding tasks, we used \texttt{text-embedding-3-small} to maintain consistency between methods.
We fix the model and the text preprocessing pipeline.
In the answer generation stage, we adopted \texttt{GPT-4o}  with a context window fixed at $1,200$ tokens per query.
\textbf{Baseline Methods.} To comprehensively evaluate the generation performance of \textsc{SG-RAG}, we compare it against ten representative retrieval-augmented or generation-based baselines: NaiveRAG~\cite{lewis2020retrieval}, RAPTOR ~\cite{sarthi2024raptor}, GraphRAG~\cite{edge2024local}, LightRAG~\cite{guo2024lightrag}, SubgraphRAG~\cite{li2025simpleeffectiverolesgraphs}, HyperGraphRAG~\cite{luo2025hypergraphragretrievalaugmentedgenerationhypergraphstructured}, LinearRAG~\cite{zhuang2025linearraglineargraphretrieval}, DyPRAG~\cite{tan2025dynamicparametricretrievalaugmented}, KAG~\cite{KAG}, GPT-5.1. 
% We additionally design two naive baselines to examine whether the performance gain of \textsc{SG-RAG} indeed stems from its structure-guided exact retrieval mechanism. 
We additionally design two naive baselines to isolate the effect of structure guided exact retrieval. A constraint wise evidence union alternative can be strong, but we omit it because its performance is sensitive to evidence aggregation under a fixed context budget.
\textit{Entity-based Retrieval} returns chunks that mention any query entity.
\textit{2-hop Graph Retrieval} expands each query entity to its two-hop neighbors.
\textbf{Evaluation Metrics.} We adopt a hybrid evaluation that combines objective metrics and subjective scoring to assess the accuracy and reasoning quality of the answer. We use four widely adopted metrics:
Precision, Recall, F1 Score, Hit@1 as the objective metrics.
We adopt \emph{Empowerment Score} (Emp.S) as \cite{guo2024lightrag}, a subjective metric, to evaluate the answers' reasoning quality (details in Appendix~\ref{appendix:emp-score}).

\begin{table}[!t]
\centering
\caption{Performance under Varying Constraints}
\small
% \begin{tabular}{@{}p{1.3cm}p{0.65cm}p{0.65cm}p{0.65cm}p{0.65cm}p{0.65cm}p{0.65cm}@{}}
\begin{tabular}{@{}>{\raggedright\arraybackslash}p{1.3cm}p{0.65cm}p{0.65cm}p{0.65cm}p{0.65cm}p{0.65cm}p{0.65cm}@{}}
\toprule
\textbf{} & \multicolumn{2}{c}{4 Constraints} & \multicolumn{2}{c}{5 Constraints} & \multicolumn{2}{c}{6 Constraints} \\
\cmidrule(lr){2-3} \cmidrule(lr){4-5} \cmidrule(lr){6-7}
& Recall & Hit@1 & Recall & Hit@1 & Recall & Hit@1 \\
\midrule
GPT-5.1         & 9.2\% & 9.0\% & 11.0\% & 10.8\% & 9.1\% & 9.0\% \\
NaiveRAG        & 16.1\% & 15.2\% & 16.2\% & 15.2\% & 15.4\% & 14.7\% \\
RAPTOR          & 29.6\% & 29.1\% & 29.5\% & 29.3\% & 28.7\% & 28.4\% \\
GraphRAG        & 37.7\% & 36.9\% & 37.6\% & 36.1\% & 36.3\% & 36.2\% \\
LightRAG        & 21.5\% & 20.0\% & 17.7\% & 17.1\% & 17.5\% & 17.0\% \\
\makecell[l]{Subgraph\\RAG}        & 42.3\% & 41.8\% & 42.5\% & 42.0\% & 39.9\% & 39.5\% \\
\makecell[l]{HyperGraph\\RAG}        & 35.8\% & 35.4\% & 33.1\% & 32.9\% & 33.3\% & 32.9\% \\
LinearRAG        & 27.1\% & 26.7\% & 26.4\% & 26.2\% & 25.8\% & 24.6\% \\
DyPRAG        & 24.9\% & 24.5\% & 25.2\% & 24.7\% & 22.6\% & 22.3\% \\
KAG        & 32.4\% & 32.3\% & 38.7\% & 37.8\% & 31.6\% & 31.0\% \\
\textbf{SG-RAG} & \textbf{72.1\%} & \textbf{71.2\%} & \textbf{69.1\%} & \textbf{68.4\%} & \textbf{70.5\%} & \textbf{69.1\%} \\
\bottomrule
\end{tabular}
\label{tab:constraint-performance}
\end{table}

\subsection{Experimental Results Analysis}
\paragraph{Performance Analysis.}
\textbf{Overall.} As shown in Table~\ref{tab:erqa-comparison}, \textsc{SG-RAG} consistently outperforms all baselines for all metrics on the three \textsc{ERQA} subsets.
For Hit@1, \textsc{SG-RAG} achieves relative gains of $34$\% (\textsc{FB-ERQA}), $102$\% (\textsc{CM-ERQA}), and $73$\% (\textsc{UD-ERQA}) over the strongest baseline, demonstrating a substantially higher probability of retrieving the correct answer on the first attempt.
For Recall, relative improvements over the strongest baseline reach $35$\%(\textsc{FB-ERQA}), $86$\%(\textsc{CM-ERQA}), and $87$\%(\textsc{UD-ERQA}), indicating substantially stronger coverage of relevant knowledge.
For F1, relative improvements over the strongest baseline of $34$\% (\textsc{FB-ERQA}), $94$\% (\textsc{CM-ERQA}), and $84$\% (\textsc{UD-ERQA}) show that \textsc{SG-RAG} achieves a more balanced trade-off between precision and recall.
Table~\ref{tab:erqa-comparison} also reports the subjective evaluation results (Emp.S) on all subsets, where GPT-5.2 is used as the evaluator. \textsc{SG-RAG} again achieves the highest Emp.S on all three subsets, exceeding the strongest baseline by $0.423$, $0.234$, and $0.333$, respectively, which indicates that it not only improves answer accuracy but also produces more coherent and informative reasoning.

\textbf{Robustness.} 
\begin{table}[!t]
\centering
\caption{Robustness under Graph Incompleteness}
\small
\begin{tabular}{lccc}
\toprule
\textbf{} & \textbf{$x=1$} & \textbf{$x=2$} & \textbf{$x=3$} \\
\midrule
Recall & 80.10\% & 79.60\% & 43.50\% \\
Hit@1  & 80.00\% & 79.20\% & 43.30\% \\
\bottomrule
\end{tabular}
\label{tab:graph-incomplete}
\end{table}
\begin{table}[!b]
\centering
\small
\caption{Performance Comparison with Naive Baselines}
\label{tab:naive}
\begin{tabular}{lccc}
\toprule
\textbf{Method} & \textbf{Precision} & \textbf{Recall} & \textbf{F1} \\
\midrule
 NaiveRAG & 8.23\% & 17.24\% & 11\% \\
 Entity-based Retrieval & 9.31\% & 16.84\% & 12\% \\
 LightRAG (1-hop) & 8.29\% & 17.00\% & 11\% \\
 2-hop Graph Retrieval & 8.71\% & 20.10\% & 12\% \\
\textbf{SG-RAG} & \textbf{60.47\%} & \textbf{61.20\%} & \textbf{61\%} \\
\bottomrule
\end{tabular}
\end{table}Table~\ref{tab:constraint-performance} reports performance under different numbers of query constraints. Since more constraints require the system to satisfy more conditions simultaneously, they represent more challenging retrieval settings. \textsc{SG-RAG} consistently remains the best-performing method across all settings. Even under six-constraint queries, \textsc{SG-RAG} still achieves a Hit@1 over $69$\%, corresponding to a $75$\% relative improvement over the strongest baseline, confirming its robustness in complex multi-constraint retrieval scenarios.
We further evaluate robustness under graph incompleteness caused by extraction errors. We randomly sample $1{,}000$ queries from \textsc{ERQA} and construct perturbed query sets by (i) deleting entities or edges from the knowledge graph evidence, or (ii) adding spurious constraints, yielding graph edit distances~\cite{gao2010survey} $x\in\{1,2,3\}$ from the original query subgraph. As shown in Table~\ref{tab:graph-incomplete}: for $x=1$ and $x=2$, Recall stays at $80.1$\% and $79.6$\%, while Hit@1 remains $80.0$\% and $79.2$\%, respectively. When the perturbation becomes more severe ($x=3$), Recall and Hit@1 drop to $43.5$\% and $43.3$\%, as a substantial portion of the original query structure is lost. Since the average query graph in this subset contains about $6$ edges, removing $3$ edges corresponds to roughly $50$\% structural loss, leaving insufficient grounding for retrieval and generation. Nevertheless, even in this severe setting, \textsc{SG-RAG} still remains above GraphRAG on \textsc{CM-ERQA} ($22.3$\% Recall and $20.2$\% Hit@1), indicating tolerance to mild graph incompleteness and graceful degradation under severe structure loss.
\textbf{Ablation.} To verify that the gains of \textsc{SG-RAG} stem from structure-guided exact retrieval, we conduct a progressive comparison on \textsc{CM-ERQA} against naive variants spanning from purely semantic retrieval to shallow neighborhood expansion, including NaiveRAG, Entity-based Retrieval, LightRAG, and 2-hop graph retrieval. 
\begin{table}[!t]
\centering
\small
\caption{Performance on the NQ Dataset}
\label{tab:single_constraint}
\begin{tabular}{lcccc}
\toprule
\textbf{Method} & \textbf{Precision} & \textbf{Recall} & \textbf{F1} & \textbf{Hit@1} \\
\midrule
NaiveRAG  & 79.51\% & 80.27\% & 79.89\% & 79.51\% \\
GraphRAG  & 85.92\% & 87.88\% & 86.89\% & 85.92\% \\
LightRAG  & 88.41\% & 89.05\% & 88.73\% & 88.41\% \\
\textbf{SG-RAG}    & \textbf{89.33\%} & \textbf{89.49\%} & \textbf{89.41\%} & \textbf{89.33\%} \\
\bottomrule
\end{tabular}
\end{table}
\begin{figure}[!b]
\centering
\includegraphics[width=0.99\columnwidth]{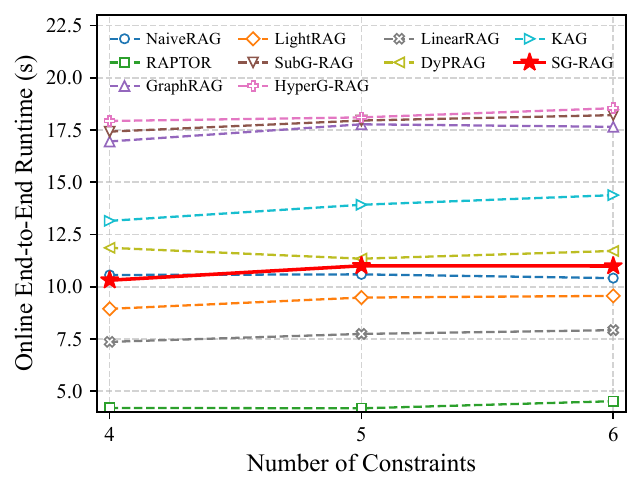} 
\caption{Online end-to-end runtime comparison.}
\label{fig-runtime}
\end{figure}
As shown in Table~\ref{tab:naive}, enlarging the retrieval scope brings only limited improvements, and these naive methods still remain in a low Precision and low F1 regime overall. 
In contrast, \textsc{SG-RAG} consistently achieves substantially higher Precision, Recall, and F1 than the naive baseline, indicating that entity co-occurrence signals and local neighborhoods are insufficient for satisfying multi-constraint queries.
\textbf{Generalization.} In addition, although \textsc{SG-RAG} is designed for multi-constraint queries, we further evaluate its robustness on single constraint queries using the NQ dataset, where queries can be represented as minimal graphs with two nodes and one edge.
As shown in Table~\ref{tab:single_constraint}, \textsc{SG-RAG} achieves the best performance, reaching $89.33$\% in Hit@1 and $89.41$\% in F1, slightly surpassing LightRAG by $0.92$ and $0.68$ points, respectively, and outperforming GraphRAG and NaiveRAG by $3.41$ and $9.82$ points in Hit@1.
These results indicate that \textsc{SG-RAG} remains stable and competitive even when the query contains one constraint.
\textbf{Case Study.} We include a case study, details are in Appendix~\ref{appendix:case study}.

\paragraph{Efficiency Analysis.}
\textbf{Online.} Figure~\ref{fig-runtime} shows the end-to-end runtimes of \textsc{SG-RAG} and the baselines under different levels of constraint. 
\textsc{SG-RAG} remains close to NaiveRAG and significantly outperforms GraphRAG. RAPTOR is the fastest, but its generation quality is the lowest, while \textsc{SG-RAG} offers a better balance between runtime and response quality. 
\textbf{Offline.} We further measure the index construction time of \textsc{SG-RAG} and other structure-aware RAG systems. 
As shown in Table~\ref{tab:offline_time}, the offline construction time is approximately $8$h for GraphRAG, $10$h for LightRAG, and $11$h for \textsc{SG-RAG} on \textsc{CM-ERQA}, indicating only a marginal overhead of about $1$ hour over LightRAG.

\begin{table}[!t]
\centering
\small
\caption{Offline index construction time across datasets.}
\label{tab:offline_time}
\begin{tabular}{lccc}
\toprule
\textbf{Dataset} & \textbf{GraphRAG} & \textbf{LightRAG} & \textbf{SG-RAG} \\
\midrule
\textsc{FB- ERQA} & 10h & 13h & 14h \\
\textsc{UD-ERQA} & 21h & 23h & 28h \\
\textsc{CM-ERQA} & 8h & 10h & 11h \\
\bottomrule
\end{tabular}
\end{table}

% \paragraph{Case Study.} We include representative queries as case study to demonstrate the reasoning capabilities of \textsc{SG-RAG}. For details, refer to Appendix H.

\begin{table}[!b]
\centering
\caption{Performance of Different Path Length $l$}
\small
\begin{tabular}{lcccc}
\toprule
\textbf{} & Hit@1 & Precision & Recall & F1 \\
\midrule
l=1 & 71.12\% & 71.14\% & 71.28\% & 71.21\% \\
l=2 & 70.45\% & 70.12\% & 70.31\% & 70.21\% \\
\bottomrule
\end{tabular}
\label{tab:path-length-effect}
\end{table}
\subsection{Sensitivity Study}

We analyze two key factors in the \textsc{SG-RAG} pipeline: path length $l$ and LLM selection.
\textbf{Path length $l$} affects query decomposition for retrieval. To ensure edge coverage, $l$ should satisfy:
$l \in \left[ \left\lceil \frac{d_q}{2} \right\rceil, d_q \right], \quad$ $d_q$ is the diameter of query graph $q$.
In \textsc{ERQA} most queries have $d_q=2$, allowing $l=1$ and $l=2$. As shown in Table~\ref{tab:path-length-effect}, accuracy remains nearly identical in both settings, since \textsc{SG-RAG} reconstructs the full subgraph prompts after matching.
\textbf{LLM Selection.}
On \textsc{FB-ERQA}, we compare six generation models in identical settings, covering representative flagship LLMs for both English and Chinese. As shown in Table~\ref{tab:generation-models}, \textsc{SG-RAG} maintains consistent quality across different LLMs, with only marginal differences in both objective and subjective metrics.

\begin{table}[!t]
\centering
\small
\caption{Performance of Different Generation Models}
\label{tab:generation-models}
\begin{tabular}{lcccc}
\toprule
\textbf{Model} & \textbf{Recall} & \textbf{Hit@1} & \textbf{Precision} & \textbf{Emp.S} \\
\midrule
GPT-5.1                 & 80.19\% & 79.32\% & 80.11\% & 2.01 \\
GPT-4o-mini           & 79.86\% & 79.21\% & 79.64\% & 1.75 \\
GLM-4V                & 80.22\% & 80.03\% & 80.17\% & 1.82 \\
GLM-4-Flash           & 78.94\% & 78.32\% & 78.64\% & 1.73 \\
Gemini-2.5            & 80.10\% & 79.56\% & 80.04\% & 1.88 \\
Qwen-3                & 79.95\% & 79.12\% & 79.80\% & 1.81 \\
\bottomrule
\end{tabular}
\end{table}

\section{Conclusion}
We introduce a novel research problem \textsc{ERP} in RAG for factual queries. To solve it, we propose \textsc{SG-RAG}, a structure guided RAG method based on subgraph matching. We construct a large-scale benchmark \textsc{ERQA} for systematic \textsc{ERP} evaluation across domains. Experiments demonstrate that \textsc{SG-RAG} consistently outperforms strong baselines by large margins across metrics.

\section*{Limitations}
\textsc{SG-RAG} is a multi-stage pipeline errors in query graph extraction, entity normalization, or label completion may propagate to downstream retrieval and generation. In particular, parts of the pipeline rely on LLM-based information extraction during knowledge graph construction, which can introduce upstream inaccuracies that may affect subsequent indexing and retrieval. In addition, our evaluation focuses on English and Chinese LLMs. We do not study other languages, and the observed trends may not directly transfer to those settings. Finally, while \textsc{SG-RAG} is designed as a general framework, we do not explore optimizations for specific domains. Tailoring components to particular vertical domains may yield further gains beyond what is reported here.
\section*{Acknowledgements}
This work was supported by the China Agricultural University ``Young Researcher'' Start-up Fund No.~QNYJY2024144 and the Visiting Scholar Program of the China Scholarship Council (CSC) No.~202506350123.

\bibliography{custom}
\clearpage
\appendix

\section{Appendix: Symbol Definitions}
\label{appendix:Symbol Definitions}

\begin{table}[!h]
\centering
{\small  % 开始9pt字体环境
\renewcommand{\arraystretch}{1.1}
\begin{tabularx}{\linewidth}{>{\raggedright\arraybackslash}p{2.1cm} p{5.9cm}}
\hline
\textbf{Symbol} & \textbf{Description} \\
\hline
$K$ & external knowledge documents \\
$G$ & knowledge graph (data graph) \\
$q^{o}$ & a user natural language query \\
$q^*$ & initial query graph \\
$q$ &  the normalized query graph \\
$a$ & the final answer generated by SG-RAG \\
$g$ & exact matched subgraphs \\
$g''$ & fallback matched subgraphs  \\
$S$ & Sets of matched subgraphs $g$\\
$d_G$ & the diameter of the $G$ \\
$d_q$ & the diameter of the $q$ \\
$I_l$ & an R*-Tree index over paths of length $l$ \\
$p_z$ ($p_q$) & path in $G$ (or $q$) \\
$p_c$ ($p'_c$) & path in candidate set \\
$v_i$ ($q_i$) & a node in $G$ (or $q$) \\
$e_{ij}$ ($e_{q_i q_j}$) & an edge in $G$ (or $q$) \\
$\mathcal{M}$ & a GNN model of $G$ \\
$Q$ & query paths before label completion \\
 $Q'$ & query paths after label completion \\
$U$ & a mapping from unknown nodes to label  \\
$P$ & fully labeled query path sets \\
$p_q.{cand\_list}$ & exact candidate paths\\
$g_{v_i}$ ($s_{v_i}$) & a star subgraph(substructure) centered at $v_i$ \\
$o(v_i)$ & a star subgraph embedding \\
$o_0(v_i)$ ($o_0(q_i)$) & a node label embedding \\
$o(p_z)$ ($o(p_q)$) & a path dominance embedding \\
$o_0(p_z)$ ($o_0(p_q)$) & a path label embedding \\
\hline
\end{tabularx}
}
\caption{Symbols and Description.}
\label{tab:symbols}
\end{table}

\section{Appendix: Graph Construction and Label Embedding Implementation Details}
\label{appendix:Graph Construction}

\paragraph{Entity and Relation Extraction Format.}
We employ a large language model (GPT-4o-mini) to extract graph components from input text, following a standardized prompt template (see Table~\ref{tab:prompt-template}). Each extracted node is formatted as a triplet: \texttt{$(\mathrm{id}(v_i), \ell(v_i), \xi(v_i))$}, and each edge as a quadruple: \texttt{$(\mathrm{id}(e_{ij}), v_i, v_j, \xi(e_{ij}))$}. This structured output ensures compatibility with downstream graph-based reasoning and retrieval modules.

\begin{table}[!h]
\centering
\renewcommand{\arraystretch}{1.3}
\small
\begin{tabular}{p{1.2cm} p{6cm}} 
\toprule
\textbf{Section} & \textbf{Content} \\
\midrule
\textbf{Goal} &
Given a paragraph, identify all semantically meaningful entities and extract structured relations among them. Each entity becomes a node and each relation becomes an edge in the output graph. \\
\addlinespace[0.6em]
\textbf{Step1: Nodes} &
Identify all entities. For each entity, extract: \par
• \texttt{$\mathrm{id}(v_i)$}: Unique node identifier \par
• \texttt{$\ell(v_i)$}: Node label (entity name) \par
• \texttt{$\xi(v_i)$}: Node description (its role or attributes) \par
\textbf{Format:} \texttt{(<$\mathrm{id}(v_i)$><|><$\ell(v_i)$><|><$\xi(v_i)$>)} \\
\addlinespace[0.6em]
\textbf{Step2: Edges} &
For each related entity pair, extract: \par
• \texttt{$\mathrm{id}(e_{ij})$}: Unique edge identifier \par
• \texttt{$v_i$}, \texttt{$v_j$}: Start and end node IDs \par
• \texttt{$\xi(e_{ij})$}: Description of the relation \par
\textbf{Format:} \texttt{(<$\mathrm{id}(e_{ij})$><|><$v_i$><|><$v_j$><|><$\xi(e_{ij})$>)} \\
\addlinespace[0.6em]
\textbf{Step3: Output Format} &
List all nodes and edges using \texttt{\#} as a separator. End with \texttt{<|COMPLETE|>}. \\
\addlinespace[0.6em]
\textbf{Input Variable} &
\texttt{Text: \{input\_text\}} \\
\bottomrule
\end{tabular}
\caption{Prompt Template for Text-to-Graph Extraction}
\label{tab:prompt-template}
\end{table}

\paragraph{Label Embedding Computation.}
We use \texttt{text-embedding-3-small} to obtain semantic embeddings of node labels. These embeddings are stored for downstream path-level embedding composition and retrieval matching.

\paragraph{Path Embedding Concatenation.}
Given a path $p_z = [v_1, v_2, \dots, v_k]$, its label embedding is defined as the concatenation of node-level embeddings:
\begin{equation}
\small
    o_0(p_z) = [o_0(v_1), o_0(v_2), \dots, o_0(v_k)] \in \mathbb{R}^{k \cdot F},
\end{equation}
preserving order-sensitive semantic representation for high-fidelity matching.

\section{Appendix: GNN-Based Dominant Embedding Training}
\label{appendix:gnn-dominant}
\begin{figure}[b]
\centering
\includegraphics[width=0.99\linewidth]{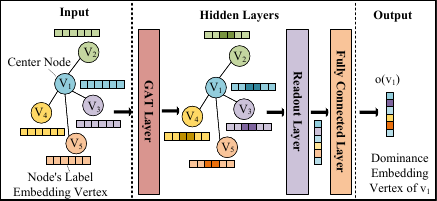}
\caption{Architecture of the GNN used for learning dominant embeddings.}
\label{fig3}
\end{figure}

\paragraph{Architecture.}
This appendix describes the training process of the GNN-based dominant embedding module used in \textsc{SG-RAG}, including GAT layer formulation, subgraph aggregation, and structural dominance constraints.
As shown in Fig~\ref{fig3}, the GNN adopts a standard GAT-based architecture with an input projection layer, an attention-based message passing layer, a readout layer, and a fully connected projection head.

\paragraph{GAT Layer Formulation.}
Given input label embeddings $x_j$ for node $v_j$, the attention coefficient between $v_i$ and $v_j$ is computed as:
\begin{equation}
\small
    a_{v_i v_j} = a(W x_i, W x_j)
\end{equation}
where $a(\cdot, \cdot)$ is a shared attention function $a: \mathbb{R}^{F'} \times \mathbb{R}^{F'} \rightarrow \mathbb{R}$, capturing the importance of $v_j$ to $v_i$. Let $N(v_i)$ denote the neighbor set of $v_i$. The attention coefficients are normalized using softmax:
\begin{equation}
\small
    \alpha_{v_i v_j} = \mathrm{softmax}(a_{v_i v_j}) = 
\frac{\exp(a_{v_i v_j})}{\sum_{v_k \in N(v_i)} \exp(a_{v_i v_k})}
\end{equation}
Then the updated node representation is:
\begin{equation}
\small
    x'_i = \sigma\left( \sum_{v_j \in N(v_i)} \alpha_{v_i v_j} \cdot W x_j \right)
\end{equation}
where $\sigma(\cdot)$ is a nonlinear activation function, and $x'_i \in \mathbb{R}^{F'}$.

\paragraph{Star Subgraph Aggregation.}
The dominant embedding of a node is computed from its local star-shaped subgraph $g_{v_i}$ via aggregation:
\begin{equation}
\small
   y_i = \sum_{v_j \in V(g_{v_i})} x'_j 
\end{equation}
\begin{equation}
\small
    o(g_{v_i}) = \sigma(W y_i), \quad o(v_i) = o(g_{v_i})
\end{equation}
\begin{figure}[t]
\centering
\includegraphics[width=0.99\linewidth]{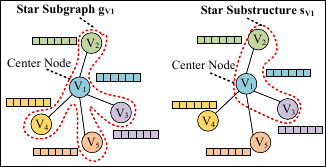}
\caption{Illustration of a star-shaped subgraph and its substructures.}
\label{fig4}
\end{figure}
As shown in Fig~\ref{fig4}, each star subgraph includes the central node and its neighbors, and all of its proper substructures are used for enforcing dominance.

\paragraph{Dominance Constraint.}
For any proper substructure $s_{v_i} \subset g_{v_i}$, the model enforces:
\begin{equation}
\small
    o(s_{v_i}) \preceq o(g_{v_i})
\end{equation}
\begin{equation}
\small
    \mathcal{L} = \sum \left\| \max(0, o(g_{v_i}) - o(s_{v_i})) \right\|_2^2
\end{equation}

This constraint ensures that smaller subgraphs embed into smaller vector regions, enabling efficient pruning.

\paragraph{Path-Level Embedding.}
Given a path $p_z = [v_1, \dots, v_k]$, its embedding is defined as:
\begin{equation}
\small
    o(p_z) = [o(v_1), \dots, o(v_k)] \in \mathbb{R}^{k \cdot d}
\end{equation}
\begin{algorithm}[!h]
\small
\caption{GNN Model Training}
\label{alg:gnn-training}
\begin{algorithmic}[1]
\REQUIRE training data $D_{train}$, learning rate $\eta$
\ENSURE trained GNN model $M$
\STATE $D_{train} \leftarrow \emptyset$
\FOR{each vertex $v_i \in V(G)$}
    \STATE extract star subgraph $g_{v_i}$ and substructures $s_{v_i}$
    \STATE add all $(g_{v_i}, s_{v_i})$ to $D_{train}$
\ENDFOR
\STATE shuffle $D_{train}$
\REPEAT
    \FOR{each batch $B \subseteq D_{train}$}
        \STATE compute embeddings and loss $\mathcal{L}(B)$
        \STATE update model: $M.\text{update}(\mathcal{L}(B), \eta)$
    \ENDFOR
    \STATE $\mathcal{L}_e \leftarrow 0$
    \FOR{each batch $B \subseteq D_{train}$}
        \STATE $\mathcal{L}_e \leftarrow \mathcal{L}_e + \mathcal{L}(B)$
    \ENDFOR
\UNTIL{$\mathcal{L}_e = 0$}
\RETURN trained model $\mathcal{M}$
\end{algorithmic}
\end{algorithm}
\begin{figure}[!b]
\centering
\includegraphics[width=0.95\linewidth]{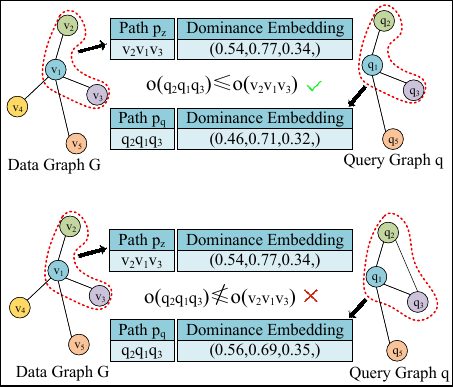}
\caption{Path-level dominant embedding matching via element-wise comparison.}
\label{fig5}
\end{figure}

\paragraph{Example Illustration.}
As shown in Fig~\ref{fig5}, \textsc{SG-RAG} verifies whether a query path $q_2 q_1 q_3$ is contained in a candidate path $v_2 v_1 v_3$ by checking:
\begin{equation}
\small
    o(q_2 q_1 q_3) \preceq o(v_2 v_1 v_3) \quad \Rightarrow \text{Valid Match}
\end{equation}
If the condition fails in any dimension, the match is rejected. This dominance check underpins the path-level pruning in the retrieval phase \textsc{SG-RAG}.

\paragraph{Training Algorithm.}
The full training process is shown in Algorithm~\ref{alg:gnn-training}. It uses a star-subgraph contrastive loss over multiple pairs to learn structure-aware embeddings that generalize to unseen queries.
\begin{table}[!h]
\centering
\renewcommand{\arraystretch}{1.3}
\small
\begin{tabular}{p{1.2cm} p{5.5cm}} % adjust for single-column width
\toprule
\textbf{Section} & \textbf{Content} \\
\midrule
\textbf{Goal} &
Given a user question, identify all semantically independent entities and extract their implicit structural relations to construct a query graph. If the question contains a target entity (i.e., the answer is unknown), it should be extracted as a node with label ``UNK''. \\
\addlinespace[0.6em]
\textbf{Step1: Nodes} &
Identify all entities. For each entity, extract: \par
• \texttt{$\mathrm{id}(q_i)$}: Unique node identifier \par
• \texttt{$\ell(q_i)$}: Node label (entity name), use ``UNK'' if it is the target to be queried \par
• \texttt{$\xi(q_i)$}: Node description (optional, can be empty if unknown) \par
\textbf{Format:} \texttt{(<$\mathrm{id}(q_i)$><|><$\ell(q_i)$><|><$\xi(q_i)$>)} \\
\addlinespace[0.6em]
\textbf{Step2: Edges} &
For each related entity pair, extract: \par
• \texttt{$\mathrm{id}(e_{q_iq_j})$}: Unique edge identifier \par
• \texttt{$q_i$}, \texttt{$q_j$}: Start and end node IDs \par
\textbf{Format:} \texttt{(<$\mathrm{id}(e_{q_iq_j})$><|><$q_i$><|><$q_j$><|>)} \\
\addlinespace[0.6em]
\textbf{Step3: Output Format} &
List all nodes and edges in order, separated by \texttt{\#}. End with \texttt{<|COMPLETE|>}. \\
\addlinespace[0.6em]
\textbf{Input Variable} &
\texttt{Text: \{user\_query\}} \\
\bottomrule
\end{tabular}
\caption{Prompt Template for Query-to-Graph Extraction}
\label{tab:query-prompt-template}
\end{table}

\section{Appendix: Entity Normalization Details}
\label{appendix:entity-normalization}

\paragraph{Entity and Relation Extraction.} The initial query graph $q^*$ is constructed by prompting a large language model to identify entities and their relationships from the natural language query $q^{o}$. The prompt format and examples are listed in Table~\ref{tab:query-prompt-template}.
\paragraph{Label Embedding.} For each node $q_i^* \in V(q^*)$, the system computes its label embedding:
\begin{equation}
\small
    o_0(q_i^*) = \text{LLM}(\ell(q_i^*)) \in \mathbb{R}^F.
\end{equation}
These embeddings serve as semantic keys for aligning query nodes with their canonical counterparts in $G$.

\paragraph{Entity Normalization via FAISS.} Due to possible lexical variations in user queries, \textsc{SG-RAG} employs a FAISS-based nearest neighbor search over the label embedding space of $G$:

\begin{itemize}[leftmargin=*, itemsep=0pt, topsep=1pt, parsep=0pt, partopsep=0pt]
    \item A deep copy of the initial graph $q^*$ is made: $q \leftarrow \text{DeepCopy}(q^*)$;
    \item Each $o_0(q_i^*)$ is used as a query vector;
    \item FAISS retrieves the most similar $q_i \in V(G)$ based on cosine similarity;
    \item The node $q_i^*$ is replaced with $q_i$ in $q$.
\end{itemize}

\paragraph{Formal Procedure.}
The normalization process is summarized in Algorithm~\ref{alg:entity-normalization}.
\begin{algorithm}[htbp]
\small
\caption{Entity Normalization}
\label{alg:entity-normalization}
\begin{algorithmic}[1]
\REQUIRE Initial query graph $q^*$ with raw entities
\ENSURE Normalized query graph $q$

\STATE $q \leftarrow \text{DeepCopy}(q^*)$
\FOR{each vertex $q_i^* \in V(q^*)$}
    \STATE $o_0(q_i^*) \leftarrow \text{Embed}(\ell(q_i^*))$
    \STATE $q_i \leftarrow \text{FAISS\_Search}(o_0(q_i^*), \text{label embedding of } V(G))$
    \STATE Replace $q_i^*$ with $q_i$ in $V(q)$
\ENDFOR
\RETURN $q$
\end{algorithmic}
\end{algorithm}

\section{Appendix: Path Decomposition and Completion}
\label{appendix:path-completion}

\paragraph{Cost-Aware Path Decomposition.}
Given a query graph $q$ and predefined path length $l$, we initialize $Q = \emptyset$ and $\text{Cost}_Q(\phi) = +\infty$. Starting from the node with highest degree, we enumerate all initial paths of length $l$ as $\text{PathSet}$, and iteratively construct the path set $\text{local}_Q$ using:
\begin{itemize}[leftmargin=*, itemsep=0pt, topsep=1pt, parsep=0pt, partopsep=0pt]
    \item minimal edge overlap with existing $\text{local}_Q$;
    \item minimal path weight $w(p) = -\sum_{q_i \in p} \deg(q_i)$.
\end{itemize}

The optimal set $Q$ minimizing $\text{Cost}_Q(\phi) = \sum w(p_q)$ is retained.

\begin{algorithm}[htbp]
\small
\caption{Cost-Model-Based Query Plan Selection}
\label{alg:cost-aware-decomposition}
\begin{algorithmic}[1]
\REQUIRE Query graph $q$, path length $l$
\ENSURE Query path set $Q$ representing the query plan $\phi$
\STATE $Q \leftarrow \emptyset$; \quad $\text{Cost}_Q(\phi) \leftarrow +\infty$
\STATE Select starting vertex $q_i$ with the highest degree
\STATE Obtain initial path set $\text{PathSet}$ of length $l$ starting from $q_i$
\FOR{each candidate initial path $p_q \in \text{PathSet}$}
    \STATE $\text{local}_Q \leftarrow \{p_q\}$; \quad $\text{local\_cost} \leftarrow 0$
    \WHILE{some edge in $E(q)$ is not covered by $\text{local}_Q$}
        \STATE Select path $p$ of length $l$ that connects with $\text{local}_Q$, minimizing:
        \begin{itemize}[leftmargin=*, itemsep=0pt, topsep=1pt, parsep=0pt, partopsep=0pt]
            \item Edge overlap with $\text{local}_Q$
            \item Path weight $w(p)$
        \end{itemize}
        \STATE $\text{local}_Q \leftarrow \text{local}_Q \cup \{p\}$
        \STATE $\text{local\_cost} \leftarrow \text{local\_cost} + w(p)$
    \ENDWHILE
    \IF{$\text{local\_cost} < \text{Cost}_Q(\phi)$}
        \STATE $Q \leftarrow \text{local}_Q$
        \STATE $\text{Cost}_Q(\phi) \leftarrow \text{local\_cost}$
    \ENDIF
\ENDFOR
\RETURN $Q$
\end{algorithmic}
\end{algorithm}

\paragraph{Label Completion for Unknown Nodes.}
For each query path $p_q$ with unknown nodes, we insert zero vectors to obtain wildcard label embeddings $o_0(p_q)$. We then traverse the R*-Tree index $I_l$ using a max-heap, comparing known dimensions of $o_0(p_q)$ with MBRs to collect aligned candidate paths $p_z$, and extract candidate label values into a mapping $U$ from node IDs to label sets.

\begin{algorithm}[!h]
\small
\caption{Find Candidate Labels for Unknown Vertices in Query Graph}
\label{alg:label-completion}
\begin{algorithmic}[1]
\REQUIRE Query path set $Q$; R*-Tree index $I_l$ over data graph $G$
\ENSURE $U$: a map from unknown vertex ID to candidate labels
\STATE $U \leftarrow \emptyset$
\STATE $\text{UnknownVertexPathset} \leftarrow \emptyset$  // Paths containing unknown nodes
\WHILE{at least one unknown vertex is not covered}
    \STATE Select a path $p_q \in Q$ that contains an uncovered unknown vertex
    \STATE $\text{UnknownVertexPathset} \leftarrow \text{UnknownVertexPathset} \cup \{p_q\}$
\ENDWHILE

\FOR{each query path $p_q \in \text{UnknownVertexPathset}$}
    \STATE Compute $o_0(p_q)$ using LLM \hfill // All-zero embedding for unknown nodes
    \STATE $\text{root}(I_l).\text{list} \leftarrow \text{UnknownVertexPathset}$
    \STATE Insert $(\text{root}(I_l), 0)$ into heap $H$
\ENDFOR

\WHILE{$H$ is not empty}
    \STATE $(N, \text{key}(N)) \leftarrow H.\text{pop()}$
    \IF{$\text{key}(N) < \min_{p_q \in \text{UnknownVertexPathset}} \|o_0(p_q)\|_1$}
        \STATE \textbf{break}
    \ENDIF
    \IF{$N$ is not a leaf node}
        \FOR{each child $N_i \in N$}
            \FOR{each query path $p_q \in N.\text{list}$}
                \IF{$o_0(p_q)$ matches $N_i.\text{MBR}_0$ on known positions}
                    \STATE $N_i.\text{list} \leftarrow N_i.\text{list} \cup \{p_q\}$
                \ENDIF
            \ENDFOR
            \IF{$N_i.\text{list} \neq \emptyset$}
                \STATE Insert $(N_i, \text{key}(N_i))$ into heap $H$
            \ENDIF
        \ENDFOR
    \ELSE
        \FOR{each stored path $p_z \in N$}
            \FOR{each query path $p_q \in N.\text{list}$}
                \IF{$o_0(p_q)$ and $o_0(p_z)$ match on known positions}
                    \STATE Extract labels from $p_z$ for unknown positions in $p_q$
                    \STATE Update $U[\text{unknown vertex id}]$ accordingly
                \ENDIF
            \ENDFOR
        \ENDFOR
    \ENDIF
\ENDWHILE

\RETURN $U$
\end{algorithmic}
\end{algorithm}

\paragraph{Enumerating Completed Paths.} Given:
\[
U = \{ q_1: [\ell_1, \ell_2],\; q_2: [\ell_3] \},
\]
we enumerate the Cartesian product of label options, e.g., $(\ell_1, \ell_3), (\ell_2, \ell_3)$, and apply each combination to a deep copy of $Q$ to obtain a label-complete set $Q'$. Repeating this for all combinations gives:
\[
P = \{ Q_1', Q_2', \ldots, Q_n'\}.
\]

This ensures that all query paths are structurally complete and semantically instantiable for path-level matching.
Algorithms~\ref{alg:cost-aware-decomposition}; Algorithms~\ref{alg:label-completion} and Algorithms~\ref{alg:enumeration} formally describe the decomposition and completion process.
\begin{algorithm}[!h]
\small
\caption{Populate the Unknown Vertices in Query Paths}
\label{alg:enumeration}
\begin{algorithmic}[1]
\REQUIRE Query path set $Q$; candidate label map $U$
\ENSURE $P \leftarrow \{Q' \mid Q'$ is a label-completed instantiation of $Q\}$

\STATE $\text{UnknownVertexIDs} \leftarrow$ keys of $U$
\STATE $\text{LabelOptions} \leftarrow [U[\text{id}] \text{ for id in UnknownVertexIDs}]$
\STATE $\text{LabelCombinations $\leftarrow$ CartesianProduct(LabelOptions)}$
\STATE $P \leftarrow \emptyset$

\FOR{each $\text{combo} \in \text{LabelCombinations}$}
    \STATE $\text{LabelMap} \leftarrow \emptyset$
    \FOR{$i = 0$ to $\text{length(UnknownVertexIDs)} - 1$}
        \STATE $\text{LabelMap}[\text{UnknownVertexIDs}[i]] \leftarrow \text{combo}[i]$
    \ENDFOR

    \STATE $Q' \leftarrow \text{DeepCopy}(Q)$
    \FOR{each path in $Q'$}
        \FOR{each vertex $v$ in path}
            \IF{$v.\text{label} = \text{NULL}$ \AND $v.\text{id} \in \text{LabelMap}$}
                \STATE $v.\text{label} \leftarrow \text{LabelMap}[v.\text{id}]$
            \ENDIF
        \ENDFOR
    \ENDFOR

    \STATE Append $Q'$ to $P$
\ENDFOR

\RETURN $P$
\end{algorithmic}
\end{algorithm}

\section{Appendix: Complexity Analysis}
\label{appendix:CA}

In this appendix, we formalize the time complexity of the SG-RAG retrieval phase. 
The retrieval phase consists of three components: entity normalization, unknown label completion, and exact path-level matching.

\paragraph{Notation.}
Let $|V(q)|$ denote the number of nodes in the query graph, $N$ the number of entities in the knowledge graph, $Q$ the set of query paths after decomposition, and $Q_u \subseteq Q$ the subset of query paths containing unknown nodes. 
Let $u$ be the number of unknown nodes, and let $t_j$ denote the number of candidate labels for the $j$-th unknown node. 
Let $h$ be the height of the R*-Tree, $f$ the average fan-out at each level, and $PP_i$ the pruning ratio at level $i$.

\begin{theorem}

The time complexity of the SG-RAG retrieval phase is
\begin{equation}
\label{eq:appendix_total_complexity}
\small
\begin{aligned}
O\Bigl(&|V(q)| \cdot \log N \\
&+ \sum_{i=0}^{h} |Q_u| \cdot f^{h-i+1} \cdot (1 - PP_i) \\
&+ \left( \prod_{j=1}^{u} t_j \right)
\cdot
\sum_{i=0}^{h} |Q| \cdot f^{h-i+1} \cdot (1 - PP_i)
\Bigr).
\end{aligned}
\end{equation}
    
\end{theorem}

\begin{proof}

We analyze the retrieval phase in three parts.
\paragraph{Entity Normalization.}
Each query node is aligned to the closest entity in the knowledge graph using FAISS-based approximate nearest neighbor search over $N$ entities. Since this search is performed once for each node in $V(q)$, the total cost of this stage is
\begin{equation}
\label{eq:appendix_entity_norm}
\small
O(|V(q)| \cdot \log N).
\end{equation}

\paragraph{Unknown Label Completion.}
Let $|Q_u|$ be the number of query paths containing unknown nodes. For each such path, SG-RAG traverses the R*-Tree to collect candidate labels. At level $i$, the traversal expands with average fan-out $f$, while only a fraction $(1-PP_i)$ of branches survives after pruning. Summing over all levels gives the total cost of unknown label completion:
\begin{equation}
\label{eq:appendix_unknown_completion}
\small
O\left(
\sum_{i=0}^{h} |Q_u| \cdot f^{h-i+1} \cdot (1 - PP_i)
\right).
\end{equation}

\paragraph{Exact Path-Level Matching.}
After candidate labels are collected, SG-RAG enumerates all completed query path sets. If the $j$-th unknown node has $t_j$ candidate labels, then the number of completed query path sets is
$
\prod_{j=1}^{u} t_j
$.
For each completed query path set, SG-RAG performs path-level retrieval over all query paths in $Q$ using the same R*-Tree traversal strategy. Therefore, the cost of this stage is
\begin{equation}
\label{eq:appendix_exact_matching}
\small
O\left(
\left( \prod_{j=1}^{u} t_j \right)
\cdot
\sum_{i=0}^{h} |Q| \cdot f^{h-i+1} \cdot (1 - PP_i)
\right).
\end{equation}

Combining Eqs.~\eqref{eq:appendix_entity_norm}, \eqref{eq:appendix_unknown_completion}, and \eqref{eq:appendix_exact_matching}, we obtain Eq.~\eqref{eq:appendix_total_complexity}.
\end{proof}

\begin{corollary}

Under the practical setting where most queries contain only one unknown node ($u = 1$) and the number of candidate labels is small ($t_1 \leq 10$), the retrieval complexity reduces to
\begin{equation}
\label{eq:appendix_practical_complexity}
\small
O\left(
|V(q)| \cdot \log N
+
(t_1 + 1)
\cdot
\sum_{i=0}^{h} |Q| \cdot f^{h-i+1} \cdot (1 - PP_i)
\right).
\end{equation}
\end{corollary}
\begin{proof}

When $u=1$, we have
$
\prod_{j=1}^{u} t_j = t_1
$.
Moreover, since $Q_u \subseteq Q$, the unknown label completion term can be upper bounded by
\begin{equation}
\small
\sum_{i=0}^{h} |Q_u| \cdot f^{h-i+1} \cdot (1 - PP_i)
\leq
\sum_{i=0}^{h} |Q| \cdot f^{h-i+1} \cdot (1 - PP_i).
\end{equation}
Substituting this bound into Eq.~\eqref{eq:appendix_total_complexity} yields
\begin{equation}
\small
O\left(
|V(q)| \cdot \log N
+
(1 + t_1)
\cdot
\sum_{i=0}^{h} |Q| \cdot f^{h-i+1} \cdot (1 - PP_i)
\right),
\end{equation}
which proves Eq.~\eqref{eq:appendix_practical_complexity}.
\end{proof}

The above result shows that, in the common case where the number of unknown nodes is small and candidate label sets remain limited, the retrieval cost of SG-RAG remains tractable in practice.

\section{Appendix: ERQA Dataset Construction and Statistics}
\label{appendix:erqa}

\begin{table}[h]
\centering
\small
\begin{tabular}{lccccc}
\toprule
\textbf{Subset} & \textbf{\#Ent.} & \textbf{\#Rel.} & \textbf{AvgE} & \textbf{AvgR} & \textbf{\#$d$} \\
\midrule
CM-ERQA       & 1.66M   & 4.58M   & 5.4 & 5.7 & 3.6 \\
FB-ERQA      & 14.5K   & 272.1K  & 6.1 & 5.4 & 3.1 \\
UD-ERQA   & 314.4K  & 53.0K   & 4.7 & 5.1 & 3.3 \\
\bottomrule
\end{tabular}
\caption{Structural Statistics of the \textsc{ERQA }Datasets}
\label{tab:sgqa-stats}
\end{table}
\paragraph{Dataset Statistics.}
As shown in Table~\ref{tab:sgqa-stats}, ``\#Ent.'' and ``\#Rel.'' indicate the number of entities and relations, ``AvgE'' and ``AvgR'' are the average number of entities and relations per query graph, and ``\#$d$'' is the average maximum path length. These statistics guide the selection of hyperparameters \textsc{SG-RAG} such as the path length $l$.

\begin{figure}[!t]
\centering
\includegraphics[width=0.98\columnwidth]{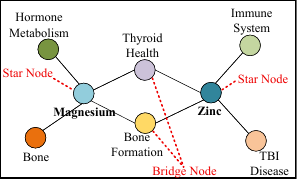}
\caption{Illustration of a bridge-star subgraph structure.}
\label{fig6}
\end{figure}

\paragraph{Bridge-Star Subgraph Construction.}
Each QA pair is constructed over a \textbf{Bridge-Star Subgraph}, defined as a pair of high-degree star nodes connected via shared bridge nodes (Figure~\ref{fig6}). Each star node forms a 1-hop subgraph, while bridge nodes enable multi-hop constraint reasoning.

\paragraph{Question Generation Process.}
For each subgraph: One star node is hidden and designated as the ground-truth answer; The visible star node's 1-hop neighbors and all bridge nodes are included in the question context.
The structured input is then converted into a natural language question via LLM prompting. The prompt format in Appendix~\ref{appendix:question-generation-prompt}.

\begin{figure}[!b]
\centering
\includegraphics[width=1\columnwidth]{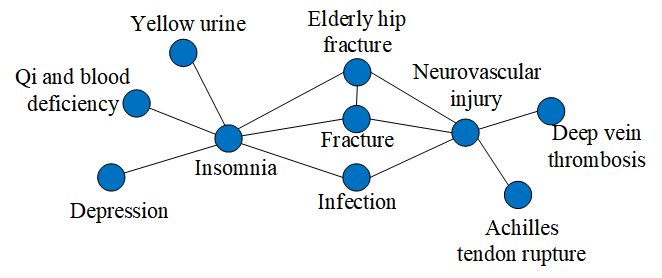} 
\caption{Subgraph structure used for query construction.}
\label{fig7}
\end{figure}

\begin{table}[!h]
\centering
\renewcommand{\arraystretch}{1.2}
\small
\begin{tabular}{p{1.8cm} p{1.6cm} p{3.4cm}}
\toprule
\textbf{Method} & \textbf{Answer} & \textbf{Key Supporting Output} \\
\midrule
GPT-5.1 & Unable determine & ``...could not identify a single disease associated with all listed complications...'' \\
NaiveRAG & Unable determine & ``...no disease found to be directly related to all the listed complications...'' \\
RAPTOR & Sleep disorder & ``Sleep disorder'' \\
GraphRAG & Breast cancer & ``...may experience thrombosis, infection, and fracture...'' \\
LightRAG & Unable determine & ``...no disease clearly causes all these complications simultaneously...'' \\
SubgraphRAG & Breast cancer & ``...may experience thrombosis, infection, and fracture...'' \\
HyperGraphRAG & Unable determine & ``...could not identify a single disease associated with all listed complications...'' \\
LinearRAG & Unable determine & ``...could not identify a single disease associated with all listed complications...'' \\
KAG & prolonged
inactivity & ``...The disease that is likely
prolonged inactivity....'' \\
\textbf{SG-RAG} & \textbf{Neurovascular injury} & ``Answer: Neurovascular injury. Reasoning: Based on the description, this condition is associated with all listed complications...'' \\
\bottomrule
\end{tabular}
\caption{Comparison of Answers}
\label{tab:case-comparison-single}
\end{table}
\paragraph{Example.}
In Figure~\ref{fig6}, “Magnesium” and “Zinc” are star nodes. “Magnesium” links to “Bone” and “Hormone Metabolism”; “Thyroid Health” and “Bone Formation” are bridge nodes. A generated question is: \emph{Which element is associated with bone formation, thyroid health, hormone metabolism, and immune system?}
Here, “Magnesium” is the correct answer, while “Zinc” is a plausible distractor.
This design supports precise multi-constraint evaluation of RAG systems.
\paragraph{Human Validation Protocol.}
 The evaluation focuses on the following three aspects.
\textbf{Fluency (1--5).} This score measures whether a query is natural, clear, and grammatically well-formed as a question. A score of $1$ indicates that the query is very unnatural or difficult to understand, while a score of 5 indicates that the query is fluent and easy to read.
\textbf{Answerability.} This metric measures whether the query can be answered based on the available knowledge graph evidence. Annotators assign ``Yes'' if the query contains enough information to support a unique factual answer, and ``No'' otherwise.
\textbf{Ambiguity.} This metric measures whether the query may reasonably correspond to more than one answer. Annotators assign ``Yes'' if multiple answers are plausible, and ``No'' if the query is sufficiently specific. We report the percentage of ambiguous queries, so lower values are better.
For cases that are difficult to judge, the annotators verify them through literature search before making a final decision. For Fluency, we report the average score across sampled queries. For Answerability and Ambiguity, we report the percentage of queries labeled as ``Yes''.

\begin{table}[!t]
\centering
\renewcommand{\arraystretch}{1.3}
\small
\begin{tabular}{p{1.4cm} p{5.5cm}}
\toprule
\textbf{Section} & \textbf{Content} \\
\midrule
\textbf{Goal} &
Given a structured subgraph and a user query, infer the most probable answer by reasoning over the entities and their relations. The answer must be one of the entities mentioned in the graph. \\
\addlinespace[0.6em]
\textbf{Step1: Input Subgraph} &
Each node contains: \par
• $id(v_i)$ \quad • $\ell(v_i)$ \quad • $\xi(v_i)$ \par
Each edge contains: \par
• $id(v_i)$ \quad •$id(v_j)$ \quad • $\ell(e_{ij})$ \\
\addlinespace[0.6em]
\textbf{Step2: Relation Paragraph} &
Convert each edge into a natural language sentence: \par
\texttt{``Node $v_i$ is related to Node $v_j$ via: [relation :$\ell(e_{ij})$].''} \par
Concatenate all such sentences into a coherent paragraph as background knowledge. \\
\addlinespace[0.6em]
\textbf{Step3: Prompt Composition} &
Combine the following parts into the final prompt: \par
• Fixed instruction explaining the task \par
• The relation paragraph from Step 2 \par
• The user’s question \par
• A constraint: the answer must be one of the mentioned entity labels. If undecidable, return unable to determine. \\
\addlinespace[0.6em]
\textbf{Formatted Example} &
\texttt{Below is a paragraph describing the relationships among entities in a structured graph.} \par
\texttt{This paragraph contains the answer to a user question. Read and reason carefully.} \par
\texttt{Note: The final answer must be one of the entity labels mentioned in the paragraph.} \par
\texttt{---} \par
\texttt{Known Relations:} \par
\texttt{Node1 is related to Node2 via: edge1.} \par
\texttt{Node1 is related to Node3 via: edge2.} \par
\texttt{Node1 is related to Node4 via: edge3.} \par
\texttt{---} \par
\texttt{User Question: Which entity is associated with multiple others?} \par
\texttt{Answer: [Entity] or ``Unable to determine''} \\
\addlinespace[0.6em]
\textbf{Input Variables} &
\texttt{Matched Subgraph: \{nodes, edges\}} \par
\texttt{User Question: \{query\}} \\
\bottomrule
\end{tabular}
\caption{Prompt Template for Answer Inference Using Matched Subgraph}
\label{tab:subgraph-answer-prompt}
\end{table}

\section{Appendix: Case Study}
\label{appendix:case study}
To further illustrate the effectiveness of \textsc{SG-RAG} in handling constraint-rich queries, we present a representative case from the \textsc{CM-ERQA} subset  :
\textbf{Query:} \textit{Which disease is likely to simultaneously cause deep vein thrombosis, acute closed Achilles tendon rupture, infection, and fracture as complications?}
\textbf{Gold Answer:} \textit{Neurovascular injury}
The entities included in this query and their neighbors are shown in Figure~\ref{fig7}. The comparison result is shown in Table~\ref{tab:case-comparison-single}

\begin{table}[!t]
\centering
\renewcommand{\arraystretch}{1.3}
\small
\begin{tabular}{p{1.4cm} p{5.7cm}}
\toprule
\textbf{Section} & \textbf{Content} \\
\midrule
\textbf{Goal} &
Generate a fluent and concise natural language question that starts with ``Which \{CORE\_TYPE\}...''. The question must simultaneously reference: \par
• Unique neighbors of the core node \texttt{\{UNIQUE\_DESC\}} \par
• Shared bridge neighbors \texttt{\{COMMON\_DESC\}} \par
This ensures that the answer must satisfy all structural constraints while avoiding ambiguity caused by bridge entities. \\
\addlinespace[0.6em]
\textbf{Step1: Placeholder Filling} &
• \texttt{\{CORE\_TYPE\}}: The type of the core entity (e.g., ``element'', ``disease'') \par
• \texttt{\{UNIQUE\_DESC\}}: Descriptions of neighbors exclusive to the core node \par
• \texttt{\{COMMON\_DESC\}}: Descriptions of shared bridge neighbors \\
\addlinespace[0.6em]
\textbf{Step2: Prompt to LLM} &
You are a Chinese language expert. Based on the placeholders provided, polish and generate a fluent, natural Chinese question that: \par
• Starts with ``Which\{CORE\_TYPE\}...'' \par
• Mentions both: \par
\texttt{\{UNIQUE\_DESC\}} (directly related features) \par
\texttt{\{COMMON\_DESC\}} (commonly co-occurring context) \par
• Do not reveal or explain the answer. Return only the question sentence. \\
\addlinespace[0.6em]
\textbf{Output Format} &
LLM should return a single sentence only, without any explanation or metadata. \\
\addlinespace[0.6em]
\textbf{Example Input} &
\texttt{\{CORE\_TYPE\} = Element} \par
\texttt{\{UNIQUE\_DESC\} = bone formation (promotes), hormone metabolism (related)} \par
\texttt{\{COMMON\_DESC\} = thyroid health (influences), immune system (associated)} \\
\addlinespace[0.6em]
\textbf{Example Output} &
\texttt{Which element is closely related to bone formation and hormone metabolism, and also influences thyroid health and participates in immune system regulation?} \\
\addlinespace[0.6em]
\textbf{Input Variables} &
\texttt{Bridge-Star Subgraph: \{nodes, edges\}} \\
\bottomrule
\end{tabular}
\caption{Prompt for Natural Question Generation}
\label{tab:question-prompt}
\end{table}

\begin{table}[!b]
\centering
\renewcommand{\arraystretch}{1.3}
\small
\begin{tabular}{p{1.3cm} p{5.9cm}}
\toprule
\textbf{Section} & \textbf{Content} \\
\midrule
\textbf{Goal} &
Automatically evaluate answers from multiple methods to the same question. The LLM acts as a reviewer and scores each answer independently according to the criteria below. \\
\addlinespace[0.6em]
\textbf{Step1:Input} &
• \texttt{Question: \{query\}} \par
• \texttt{Gold Answer: \{gold\_answer\}} \par
• \texttt{Answer List:} \par
\hspace{1em}[Method\_A] \{answer\_A\} \par
\hspace{1em}[Method\_B] \{answer\_B\} \par
\hspace{1em}\dots \par
\hspace{1em}[Method\_F] \{answer\_F\} \\
\addlinespace[0.6em]
\textbf{Step2: Evaluation Instruction} &
You are a senior evaluator. Please carefully read the question, the gold-standard answer, and the list of candidate answers. For each answer, assign scores based on the criteria below. Return the evaluation as a structured JSON. \\
\addlinespace[0.6em]
\textbf{Scoring Criteria} &
• \textbf{Logical Coherence (0–2)}: Is the reasoning clear, complete, and well-sequenced? \par
• \textbf{Insight (0–1)}: Does the answer offer new insight or helpful suggestions? \\
\addlinespace[0.6em]
\textbf{Output Format} &
Return the scores as follows: \par
\texttt{\{} \par
\hspace{1em}\texttt{"Method\_A": \{"logic": L1, "insight": I1, "total": T1\},} \par
\hspace{1em}\texttt{"Method\_B": \{"logic": L2, "insight": I2, "total": T2\},} \par
\hspace{1em}\texttt{\dots} \par
\hspace{1em}\texttt{"Method\_F": \{"logic": L3, "insight": I3, "total": T3\}} \par
\texttt{\}} \\
\bottomrule
\end{tabular}
\caption{Prompt Template for Subjective Evaluation}
\label{tab:llm-subjective-eval}
\end{table}

\paragraph{Analysis of Methods.}
Below we provide detailed observations for each method's performance:

\begin{itemize}[leftmargin=*, itemsep=0pt, topsep=1pt, parsep=0pt, partopsep=0pt]
  \item \textbf{GPT-5.1} generates generalized scenarios (e.g., trauma, diabetes) with no precise answer, failing to enforce multi-constraint reasoning.
  \item \textbf{NaiveRAG} retrieves texts related to individual entities, but lacks a mechanism to ensure global constraint satisfaction.
  \item \textbf{RAPTOR} retrieves some relevant candidates but is affected by noisy context, leading to inaccurate answer selection.
  \item \textbf{GraphRAG} partially matches constraints , but ignores ``Achilles tendon rupture,'' resulting in hallucination of ``breast cancer.''
  \item \textbf{LightRAG} covers all constraint terms but fails to reason over their intersection due to lack of co-occurrence modeling.
  \item \textbf{SubgraphRAG} capturing only part of the evidence subgraph and overemphasizing locally relevant complications, which again leads to the incorrect answer ``breast cancer.''
  \item \textbf{HyperGraphRAG} fails to identify a disease satisfying all listed complications and therefore returns an indeterminate answer.
  \item \textbf{LinearRAG} also fails to produce a specific answer, indicating that although it may retrieve semantically related information efficiently.
  \item \textbf{KAG} produces ``prolonged inactivity,'' which is not even a disease entity, indicating that it is distracted by a loosely related local clue rather than identifying the target condition .
  \item \textbf{SG-RAG} successfully reconstructs the full constraint structure, retrieves a matching subgraph, and generates a precise answer.
\end{itemize}

\section{Appendix: Prompt Template for Answer Inference from Matched Subgraph}
\label{appendix:subgraph-inference-prompt}

This prompt guides the model to infer an answer from a matched subgraph by reasoning over structured relations and answering a user’s natural language query. The model is instructed to restrict its final answer to one of the entities explicitly mentioned in the graph, the detail shown in Table~\ref{tab:subgraph-answer-prompt}.

\section{Appendix: Prompt Template for Question Construction}
\label{appendix:question-generation-prompt}

This prompt is used to generate a natural and concise question based on a bridge-star subgraph. The question must start with “Which \{CORE\_TYPE\}...” and mention both the **unique neighbors** of the core node and the **bridge neighbors**, ensuring the structural constraints are embedded and bridge nodes are disambiguated, the detail shown in Table\ref{tab:question-prompt}.

\section{Appendix: matching algorithm}
\label{appendix:exact subgraph matching algorithm}

\begin{algorithm}[!h]
\small
\caption{Exact Subgraph Matching with GNN-based Path Dominance Embedding}
\label{alg:path-matching}
\begin{algorithmic}[1]
\REQUIRE query graph $q$; fully labeled query path sets $P$; trained GNN model $M$; R*-Tree $I_l$ over data graph $G$.
\ENSURE exact match subgraph set $S$.

\STATE $S \leftarrow \emptyset$
\FOR{each query path set $Q_i' \in P$}
    \FOR{each query path $p_q \in Q_i'$}
        \STATE $p_q.\text{cand\_list} \leftarrow \emptyset$
        \STATE Obtain $o(p_q)$ via GNN
        \STATE Obtain $o_0(p_q)$ via LLM
    \ENDFOR

    \STATE $\text{root}(I_l).\text{list} \leftarrow Q_i'$
    \STATE Insert $(\text{root}(I_l), 0)$ into heap $H$

    \WHILE{$H$ is not empty}
        \STATE $(N, \text{key}(N)) \leftarrow H.\text{pop}()$
        \IF{$\text{key}(N) < \min_{p_q \in Q_i'} \| o(p_q) \|_1$}
            \STATE \textbf{break}
        \ENDIF

        \IF{$N$ is an internal node}
            \FOR{each child $N_i \in N$}
                \FOR{each $p_q \in N.\text{list}$}
                    \IF{$o_0(p_q) \in \text{MBR}_0(N_i)$ \AND $DR(o(p_q)) \cap \text{MBR}(N_i) \neq \emptyset$}
                        \STATE $N_i.\text{list} \leftarrow N_i.\text{list} \cup \{p_q\}$
                    \ENDIF
                \ENDFOR
                \IF{$N_i.\text{list} \neq \emptyset$}
                    \STATE Insert $(N_i, \text{key}(N_i))$ into $H$
                \ENDIF
            \ENDFOR
        \ELSE
            \FOR{each $p_z \in N$}
                \FOR{each $p_q \in N.\text{list}$}
                    \IF{$o_0(p_q) = o_0(p_z)$}
                        \IF{$o(p_q) \preceq o(p_z)$}
                            \STATE $p_q.\text{cand\_list} \leftarrow p_q.\text{cand\_list} \cup \{p_z\}$
                        \ENDIF
                    \ENDIF
                \ENDFOR
            \ENDFOR
        \ENDIF
    \ENDWHILE

    \STATE Assemble $S$ from all $p_q.\text{cand\_list}$ using Algorithm 7
\ENDFOR
\RETURN $S$
\end{algorithmic}
\end{algorithm}

\section{Appendix: Empowerment Score Definition and Evaluation}
\label{appendix:emp-score}
To assess the reasoning quality of generated answers, we adopt a subjective evaluation metric termed \textbf{Empowerment Score(Emp.S)}, designed to approximate human judgment. Each answer is scored by an LLM evaluator along two dimensions:

\begin{itemize}[leftmargin=*, itemsep=0pt, topsep=1pt, parsep=0pt, partopsep=0pt]
    \item \textbf{Logical Coherence ($0$--$2$ points)}: Does the response follow a clear and structured reasoning process?
    \item \textbf{Informational Value ($0$--$1$ point)}: Does the response provide useful insights, elaboration, or interpretative depth?
\end{itemize}

The total score ranges from $0$ to $3$.
All answers from \textsc{SG-RAG} and baseline systems are assessed under the same prompt setting by the same LLM judge to ensure fairness. We report average Empowerment Scores across pairwise comparisons. Evaluation prompt format and scoring rubric are included in Table~\ref{tab:llm-subjective-eval}.

\section{Appendix: assembly algorithm}
\label{appendix:assembly algorithm}
\begin{algorithm}[!h]
\small
\caption{Assemble Subgraphs}
\label{alg:subgraph-assembly}
\begin{algorithmic}[1]
\REQUIRE Query graph $q$;
 Fully labeled query path set $Q'$;\ Each $p_{q_i}$ is associated with $p_{q_i}.\text{cand\_list}$
\ENSURE Exact subgraph set ${S}$
\STATE ${S} \leftarrow \emptyset$
\STATE $F \leftarrow p_{q_1}.\text{cand\_list} \times p_{q_2}.\text{cand\_list} \times \cdots \times p_{q_k}.\text{cand\_list}$
\FOR{each combination $\{p_{f_1}, \dots, p_{f_k}\} \in F$}
    \STATE $g \leftarrow$ empty graph
    \STATE $\text{query\_node\_map} \leftarrow$ empty map
    \STATE $\text{conflict} \leftarrow \textbf{False}$
    \FOR{each path $p_f$ in combination}
        \FOR{each query node $v_f$ in $p_f$}
            \STATE $v_s \leftarrow$ mapped node of $v_f$
            \IF{$v_f \in \text{query\_node\_map}$}
                \IF{$\text{query\_node\_map}[v_f] \neq v_s$}
                    \STATE $\text{conflict} \leftarrow \textbf{True}$; \textbf{break}
                \ENDIF
            \ELSE
                \STATE $\text{query\_node\_map}[v_f] \leftarrow v_s$
                \STATE add $v_s$ to $g$ if not present
            \ENDIF
        \ENDFOR
        \STATE add all edges in $p_f$ to $g$
    \ENDFOR
    \IF{not $\text{conflict}$}
        \STATE $S \leftarrow S \cup \{g\}$
    \ENDIF
\ENDFOR
\RETURN ${S}$
\end{algorithmic}
\end{algorithm}

\end{document}